\documentclass{ipi}

\usepackage[T1]{fontenc} % Use 8-bit encoding that has 256 glyphs
\usepackage[english]{babel} % English language/hyphenation
\usepackage{amsmath, amsfonts, amsthm} % Math packages
\usepackage{cite}

\usepackage[
  colorlinks=true,urlcolor=blue,citecolor=red,linkcolor=red,bookmarks=true]{hyperref}

\usepackage{color}
\usepackage{pgfplots}
\usepackage{tikz}
\usepackage{graphicx}
\usepackage{paralist}
\usepackage{graphics} %% add this and next lines if pictures should be in esp format
\usepackage{epsfig} %For pictures: screened artwork should be set up with an 85 or 100 line screen
\usepackage{graphicx} 
\usepackage{epstopdf} 
\usepackage[colorlinks=true]{hyperref}
\hypersetup{urlcolor=blue, citecolor=red}

\def\currentvolume{X}
 \def\currentissue{X}
  \def\currentyear{20xx}
   \def\currentmonth{XX}
    \def\ppages{X--XX}
     \def\DOI{10.3934/ipi.xx.xx.xx}

\newtheorem{theorem}{Theorem}[section]

\newtheorem{proposition}{Proposition}

\theoremstyle{definition}
\newtheorem{definition}[theorem]{Definition}

\newcommand{\bs}{\boldsymbol}
\allowdisplaybreaks[3]

\newcommand{\var}{\text{Var}}
\newcommand{\corr}{\text{Corr}}
\newcommand{\cov}{\text{Cov}}

\newcommand{\dom}{\mathcal{D}om}

\newcommand{\precop}{\mathcal{L}}
\newcommand{\Op}{\mathcal{A}}
\newcommand{\covop}{\mathcal{C}}

\newcommand{\n}{\boldsymbol{n} }
\newcommand{\dn}{\partial \n }
\newcommand{\x}{{\boldsymbol{x}}}
\newcommand{\y}{{\boldsymbol{y}}}
\newcommand{\z}{{\boldsymbol{z}}}

\newcommand{\xmy}{ \| \x - \y \| }

\newcommand{\kr}{\kappa r}
\newcommand{\numberthis}{\addtocounter{equation}{1}\tag{\theequation}}

\newcommand{\mycite}{\cite}
\DeclareMathOperator*{\argmin}{arg\,min}

\title[Boundary Conditions and PDE-Based Covariance Operators]
{Mitigating the Influence of the Boundary on
  PDE-based Covariance Operators}

\author{Yair Daon \and Georg Stadler}

\subjclass{Primary: 
  62F15, % Statistics - bayesian inference
  35R30, % PDE - Inverse problems
  65C50; % Numerical analysis - other computaitonal problems in probability
  Secondary:
  28C20, % Measure and integration - set functions and measures and integrals in
}
\keywords{Gaussian random fields, Mat\'ern kernels, boundary
conditions, fast PDE solvers,
Bayesian statistics, Inverse problems.}

\thanks{Supported in part by the National Science Foundation under
  grants \#1507009 and \#1522736, and by 
the U.S.\ Department of Energy Office of Science, Advanced
Scientific Computing Research (ASCR), Scientific Discovery through
Advanced Computing (SciDAC) program. }

\begin{document}
\maketitle

\vspace*{-2ex}
{\footnotesize 
 \centerline{Courant Institute, New York University}
   \centerline{ New York, NY 10012, USA}
}

\bigskip

\begin{abstract}
  Gaussian random fields over infinite-dimensional Hilbert spaces
  require the definition of appropriate covariance operators. The use
  of elliptic PDE operators to construct covariance operators allows
  to build on fast PDE solvers for manipulations with the resulting
  covariance and precision operators. However, PDE operators require a
  choice of boundary conditions, and this choice can have a strong and
  usually undesired influence on the Gaussian random field. We
  propose two techniques that allow to ameliorate these boundary
  effects for large-scale problems. The first approach combines the
  elliptic PDE operator with a Robin boundary condition, where a
  varying Robin coefficient is computed from an optimization problem.
  The second approach normalizes the pointwise variance by rescaling
  the covariance operator. These approaches can be used individually
  or can be combined. We study properties of these approaches, and
  discuss their computational complexity. The performance of our
  approaches is studied for random fields defined over simple and
  complex two- and three-dimensional domains.
\end{abstract}

\section{Introduction}
Gaussian random fields over functions, sometimes referred to as
continuously indexed Gaussian random fields, are important in spatial
statistical modeling, geostatistics and in inverse problems. They are
described through a mean and a covariance operator, usually defined
over a Hilbert space.  Efficient manipulation of random fields is of
critical importance in applications. In particular, one commonly
requires the application of the covariance operator and of its
inverse, the precision operator, to vectors from the function
space. Additionally, computation of realizations from the distribution
requires the ability to apply a square root of the covariance operator
to vectors.

Constructing covariance operators from elliptic PDE operators, which
has recently gained popularity \mycite{Stuart10,
  Bui-ThanhGhattasMartinEtAl13, LindgrenRueLindstroem11,
  RoininenHuttunenJanneEtAl14, IsaacPetraStadlerEtAl15}, allows
one to build on available fast PDE solvers for the required manipulations.
This leads to a correspondence between domain Green's functions of
PDE operators and covariance functions of the Gaussian random
fields. On bounded domains, PDE operators require the definition of
boundary conditions, which has implications for the resulting
covariance operators.
Namely, this can lead to increased/decreased correlation and pointwise
variance close to the boundary, which is usually undesirable
from a statistical perspective. This behavior is illustrated in
figure~\ref{fig:problem illustration} and has also been observed in
\mycite{Bui-ThanhGhattasMartinEtAl13, LindgrenRueLindstroem11,
  RoininenHuttunenJanneEtAl14}. 
In this work, we present methods to ameliorate these boundary effects.
Since we aim at large-scale problems, we are interested in scalable
optimal complexity algorithms that avoid dense matrix
operations or matrix assembly.  Our target is to find domain Green's functions that are
as similar as possible to the free-space Green's functions of the
precision operator, which are Mat\'ern covariance functions.  We
present two methods towards achieving this objective.

The first method combines the PDE operator with a
homogeneous Robin boundary condition 
$\beta u +\frac{\partial u}{dn} = 0$,
with a varying Robin coefficient 
$\beta=\beta(\x)$. This coefficient function is derived as solution to an
optimization problem that aims at making the difference between the
domain and the free-space Green's functions small. Our approach exploits
the definition of the domain Green's function and uses the fact that explicit
expressions for the free-space Green's functions are available or
can easily be computed numerically. 
For one-dimensional domains, $\beta$ can be chosen such that the
effect of boundary conditions vanishes completely. For two- and
three-dimensional domains, $\beta$ can be chosen to minimize boundary
effects in an averaged sense. The approach only requires computation
of inner products and is thus feasible for large-scale problems.

The second method we propose amounts to a rescaling of
the covariance operator $\covop$ that is constructed from elliptic PDE
operators. It can be combined with the approach discussed above.
This rescaled operator has constant pointwise variance (a property that
$\covop$ above does not have).
The idea is most easily understood in finite dimensions:
For a covariance matrix $\Sigma$, with diagonal $D_{ij} :=
\Sigma_{ij} \delta_{ij}$, the rescaled matrix $\Sigma'=D^{-\frac12}
\Sigma D^{-\frac12}$ is also a covariance matrix, and it has a constant
unit diagonal.

\begin{figure}
  \minipage{0.52\textwidth}
  \begin{tikzpicture}[thick,scale=.85, every node/.style={scale=0.99}]
    \begin{axis}
      [
      xmin = 0,
      xmax = 0.5,
      xlabel = {$s$},
      ylabel = {$c(\x (s), \x^\star)$},
      ymin   = 0,
      compat = 1.3,
      % ymax   = 130,
      ytick = \empty,
      legend cell align=left
      ]
      \draw[black!30!white, thin] (50,0) -- (50,130);
      \addplot [thick, black, mark=none] table {square_Free_Space_Greens_Function.txt};
      \addlegendentry{\large Free-Space};
      \addplot[thick, blue, mark=none, dashed] table {square_Dirichlet_Greens_Function.txt};
      \addlegendentry{\large Dirichlet BC};
      \addplot[thick, red, dotted, mark=none] table {square_Neumann_Greens_Function.txt};
      \addlegendentry{\large Neumann BC};
      % \node at (0.85cm,0.2cm) {$\y$};
    \end{axis}
    \begin{axis}
      [
      compat=1.3,
      axis lines = none,
      xmin = -0.1,
      xmax = 1.1,
      ymin   = -0.1,
      ymax   = 1.1,
      xtick = \empty,
      ytick = \empty,
      height = 4.5cm,
      width = 4.5cm,
      at={(3.9cm,1.2cm)},
      legend cell align=left
      ]
      \addplot [thick, black!50!white, mark=none, fill=black!10!white] table {square_vertices.txt};
      \addplot [thick, red  , mark=none] table {square_line.txt};
      \addplot [only marks, mark = *, mark size=1.2] table {square_source.txt};
      \node at (2.4cm,0.5cm) {$\Omega$};
      \node at (0.6cm,1.8cm) {$\x^\star$};
      \draw [black, thin] (0.35cm,1.45cm) -- (0.45cm,1.65cm);
      \draw [black, thin] (0.6cm,1.45cm) -- (0.8cm,1.2cm);
      \node at (1.4cm,1.1cm) {cross section};
    \end{axis}
  \end{tikzpicture}
  \endminipage\hfill
  \minipage{0.48\textwidth}
  \includegraphics[width=\linewidth, trim=0 0 0cm 0, clip]{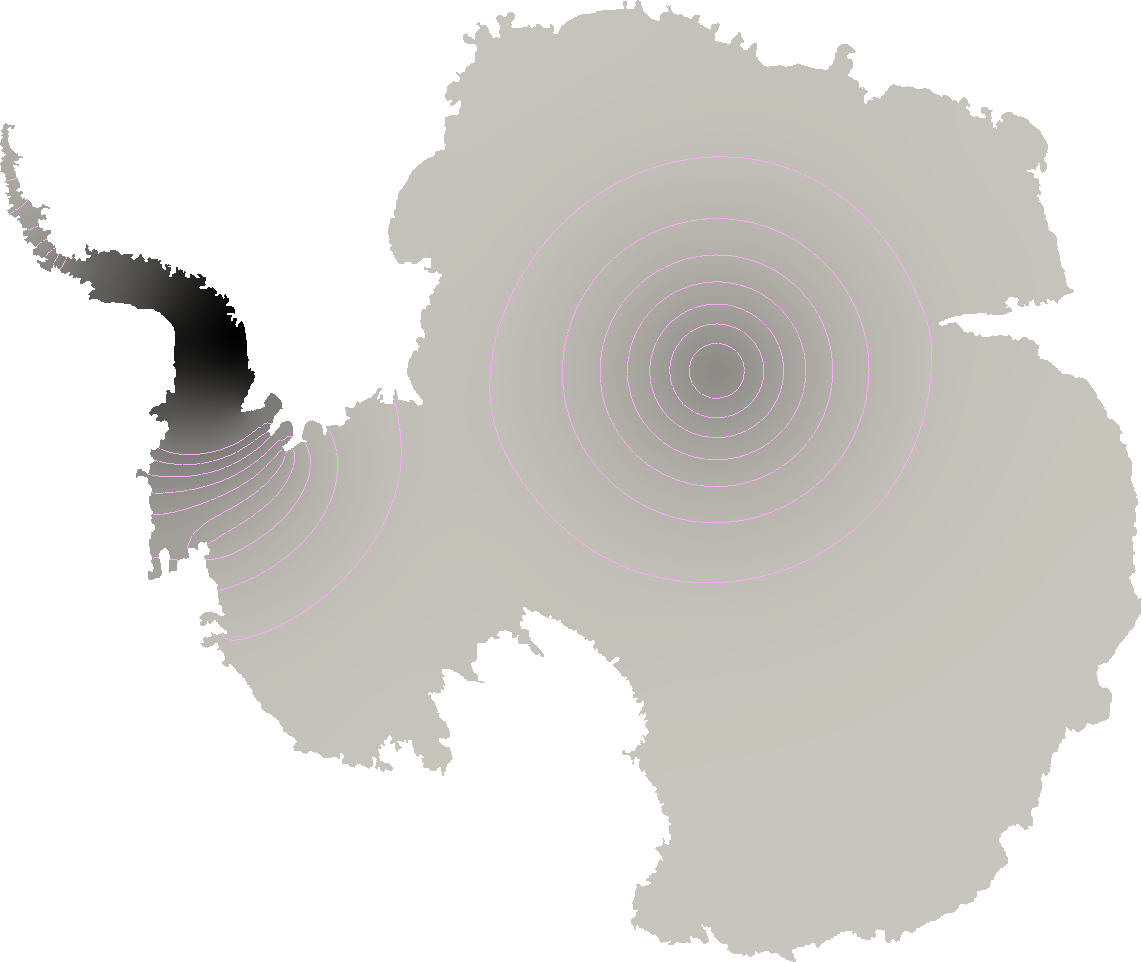} 
  \endminipage\hfill
  \caption{Left: Cross sections through covariance functions induced
    by elliptic PDE operators with different boundary
    conditions. Shown is also a sketch of the domain $\Omega=[0,1]^2$
    and the cross section $\x(s) = (s, 0.5)^T$.  The center is located
    at $\x^\star = \x(0.05) = (0.05, 0.5)^T$.  Right: Two covariance
    functions on the Antarctica domain (see
    Sec.~\ref{subsec:Antarctica}). The magnitude of the left
    covariance function exceeds the gray scale used to show the
    covariance between the centers and the points of the domain. The
    discrepancy between the covariance is due to the use of Neumann
    boundary conditions for the differential operator.}
  \label{fig:problem illustration}
\end{figure}
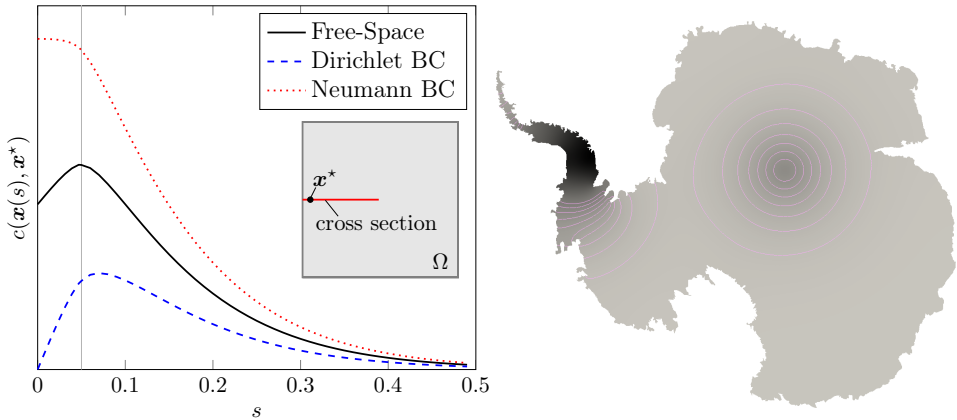

\subsection*{Related Work}
In spatial statistics, the use of covariance operators is motivated 
by the need for fast computations \mycite{SimpsonLindgrenFinnEtAl12s,
SimpsonLindgrenFinnEtAl12}. The connection
between inverse elliptic operators and Gaussian fields was originally
established in \mycite{Whittle63}. Building on this connection 
and results in \mycite{Besag81}, the authors
of \mycite{LindgrenRueLindstroem11} show that discretizing an inverse
elliptic covariance operator is valid, from a statistical
perspective. This results in a (discretely-indexed) Gaussian field
with a sparse precision operator due to the locality of
differential operators. This sparsity allows for fast application of
the precision. Fast application of the covariance operator is possible
building on fast elliptic PDE solvers.

A parallel approach aiming at Bayesian inverse problems was
established by Stuart \mycite{Stuart10}. Contrary to
\mycite{LindgrenRueLindstroem11}, the author's motivation is to develop
the theory of Bayesian inference in function spaces.  The advocated
approach (which we follow) is that all algorithms should be
presented and studied in function spaces, i.e., in infinite
dimensions. Taking this approach, the author is lead to the use of
``Laplacian-like'' precision operators \mycite[Assumption
  2.9]{Stuart10}, which are used to define Gaussian priors for
Bayesian inverse problems. In some respect, \mycite{Stuart10} and
\mycite{LindgrenRueLindstroem11,SimpsonLindgrenFinnEtAl12s} draw similar
conclusions. They argue that using covariance operators is
superior to using covariance functions, both from a theoretical
as well as from a computational perspective.

The role of PDE-operator boundary conditions if the domain $\Omega
\subseteq \mathbb{R}^d$ is bounded was already observed to cause
variance inflation near the boundary in
\mycite{LindgrenRueLindstroem11}. To avoid this effect, domain
extension was proposed in
\mycite{LindgrenRueLindstroem11,SimpsonLindgrenFinnEtAl12s,RoininenHuttunenJanneEtAl14}.
As an alternative and related to one of the methods we propose in this
paper, in \mycite{RoininenHuttunenJanneEtAl14} the authors propose to
use a Robin boundary condition of the form $\beta u + \frac{\partial
  u}{\partial n} = 0$. They conduct numerical experiments to
empirically find a constant, boundary effect mitigating coefficient
$\beta$ in the Robin condition for a two-dimensional circular
domain. In \mycite{CalvettiKaipioSomersalo06} the authors suggest
sampling values on the boundary according to the correct distribution
and then using these values as Dirichlet data for the domain. This
approach is technical in higher dimensions and it requires assembled
matrices.

\subsection*{Contributions} The main contributions of this work are as
follows: (1) The proposed methods mitigate boundary
effects arising in continuously indexed Gaussian random fields when
elliptic PDE operators are used to construct covariance operators.
They are computationally feasible, do not require assembled
matrices, nor the extension of the computational domain.
(2) We present simple and fast algorithms for the
approximation of the quantities used in our methods. Once these
upfront computations, which depend on the domain and the PDE
operator, are available, all remaining computations (covariance and precision
application, and computation of samples) are as efficient as in the
original method.
(3) We perform a comprehensive numerical study of the proposed
methods on simple as well as complex geometries, such as the Antarctica
domain from \mycite{IsaacPetraStadlerEtAl15}.

\subsection*{Limitations} We also remark limitations of our methods.
(1) To compute the optimal Robin coefficient, an integration over the
domain must be performed for each point on the boundary, where the Robin coefficient is
needed.\footnote{These points depend on the
numerical method used. For the finite element method, for instance,
values of the Robin coefficient are usually needed at boundary
quadrature points.} However, this integration can be accelerated by realizing
that in many cases the Robin coefficient varies smoothly, and hence
one may use interpolation between adjacent points.  Moreover, the
integrands decay rapidly and thus the integration can be restricted to
a small part of the domain. (2) Computation of the integrals in
the Robin method can be challenging due to the singularity 
of the integrands. As a remedy, we discuss approximations 
that allow computation of these integrals at an
accuracy that suffices for our purposes.
(3) For the variance normalization method, we
require knowledge of the pointwise variance over the
domain. Fortunately, this field is often smooth and thus interpolation
from a few points to the entire domain is possible.  Additionally, one
can leverage potential symmetries in the geometry to speed up the
computation of the pointwise variance.
(4) The upfront computations our methods require depend on the PDE
operator used to define the covariance operator. If one uses
a hierarchical method in which the PDE
operator varies, the proposed approach can become computationally
expensive.

\section{Preliminaries}\label{section:preliminaries}
Let $\Omega \subseteq \mathbb{R}^d$, $d=1,2,3$ be a bounded open
domain with piecewise smooth boundary $\partial \Omega$.  Throughout
this paper, we are concerned with Gaussian measures over spaces of
functions defined on $\Omega$. We first recall definitions of
Gaussian measures and Gaussian fields---see, e.g., \mycite{Hairer09,Prato06} for
details. 

\subsection{Gaussian measures}\label{subsec:gaussian measures}
Let $\mu$ a measure on a separable Banach space $X$ and $u \sim
\mu$. We say $\mu$ is Gaussian if $\forall \ell\in X^{*}$, if there
exist $m_{\ell}$ real and $\sigma_{\ell}$ non-negative such that
$\ell(u)\sim \mathcal{N}( m_{\ell}, \sigma_{\ell} )$ is
Gaussian. Consider $X = C(\Omega)$, the space of continuous functions
on $\Omega$ with the $\sup$-norm, so that $X \subseteq L^2(\Omega)$.
Taking this view, one can specify a Gaussian measure $\mu$ on $X$ by
first taking a mean $m \in X$ and a (linear) self-adjoint positive
definite trace class covariance operator $\covop: L^2(\Omega)
\to L^2(\Omega)$. Since samples $\mathcal{N} (m ,\covop)$ are
continuous for the choices of $\covop$ we consider below, $X$
has full measure and by \mycite[Ex. 3.39]{Hairer09} we have a Gaussian
measure on $C(\Omega)$. We still write $\mathcal{N}(m, \covop)$ for the
corresponding Gaussian measure on $X$. If $h\in L^2(\Omega)$ is
discontinuous, then $\covop^{-1} h $ is empty and $\langle h,
\covop^{-1} h \rangle = \| \covop^{-\frac12}h \| = \infty$,
\emph{informally} making the likelihood of observing $h$ zero. Thus,
the Gaussian measure gives full measure to $X$.

\subsection{Gaussian random fields}\label{subsec:gaussian fields}
For our purposes, a Gaussian random field is a random function
$u:\Omega \to \mathbb{R}$ such that for all finite sets $\{ \x_i
\}_{i=1}^n \subseteq \Omega$, the random vector $(u(\x_1), ... ,
u(\x_n))^{T}$ is a multivariate normal. For simplicity, here we
consider a centered field, i.e., $m(\x) := \mathbb{E}[u(\x)]\equiv 0$.
The corresponding covariance function $c: \Omega \times \Omega \to
\mathbb{R}$ is defined as $c(\x,\y) := \mathbb{E}[u(\x)u(\y)]$. A
covariance function can also be used as a kernel for an integral
operator. The resulting operator is given by
\begin{align*}
  (\covop f )(\x) = \int_{\Omega} c(\x,\y) f(\y) d\y.
\end{align*}
If $c$ is positive-definite \mycite{Varadhan01}, $\covop$ is a valid
covariance operator in the sense of section \ref{subsec:gaussian
  measures}.  This connection motivates considering $\covop$ to be an
inverse elliptic operator, making $c$ the Green's function of that
operator.  In this case, writing $c(\x, \y) = (\covop \delta_{\x}
)(\y)$ is well-defined from a PDE perspective and we use this identity
below.  Now, the connection with Gaussian measures is
straightforward---a Gaussian random field defines a Gaussian measure
on $C(\Omega)$.

\subsection{Inverse elliptic covariance operators}
On $\Omega$, consider the elliptic differential operator 
\begin{equation}\label{eq:operator}
\Op := -\gamma \Delta + \alpha
\end{equation}
with constants $\gamma, \alpha > 0$. The domain on which $\Op$ is
defined depends on the choice of boundary conditions. We will discuss
different domains $\dom(\Op)$ and the implied properties for
covariance operators derived from $\Op$.  We assume that $\Omega$ is
such that $\Op$ is a Laplacian-like operator in the sense of
\mycite[Assumption 2.9]{Stuart10} when equipped with homogeneous
Dirichlet, Neumann or Robin boundary conditions. The operator
$\Op^{-p}$ is a valid covariance operator for $p>d/2$, with samples
that are $s$-H\"older continuous for all $s < \min \{1,p-d/2\}$. The
covariance function of the free-space operator has a characteristic
length of $\sqrt{8(p-d/2} \sqrt{\gamma/ \alpha}$ meaning that at that
distance away from a source $\x$, the covariance decays to $0.1$ of
its maximal value (attained at $\x$)
\mycite{LindgrenRueLindstroem11}. Specifically, $\Op^{-1}$ is a
covariance operator for $d=1$ and $\Op^{-2}$ is a covariance operator
for $d=1,2,3$ \mycite{Stuart10}. The boundary conditions of $\Op^2$
are inherited from the boundary conditions of $\Op$, which we denote
by $\text{BC}(\cdot)=0$, i.e., $u = \Op^{-2}f$ is defined as the
solution of the mixed system
\begin{align}
  \begin{split}
    \Op v 
    &= f \qquad
    \text{ in } \Omega, \\
    \text{BC}(v)
    &= 0 \qquad
    \text{ on } \partial \Omega, \\
    \Op u
    &= v \qquad
    \text{ in }\Omega, \\
    \text{BC}(u)
    &= 0 \qquad
    \text{ on } \partial \Omega.
  \end{split}
\end{align}
This implies that $\Op^{-1}$ is a square root of $\Op^{-2}$. This
choice of boundary conditions renders sampling from a centered
Gaussian with covariance operator $\Op^{-2}$, $\mathcal{N}(0,
\Op^{-2})$, straightforward. Namely, samples are generated as $u \sim
\Op^{-1}\mathcal{W}$ where $\mathcal{W}$ is white noise, and this can
be interpreted in infinite dimensions---see
\mycite{LindgrenRueLindstroem11, Whittle63}.  The key property is, as
hinted in section \ref{subsec:gaussian fields}, that the Green's
function, $G_p$, of $\Op^{p}$ with appropriate boundary conditions is
the covariance function of a Gaussian measure with covariance operator
$\Op^{-p}$.  Specifically, let $u \sim \mathcal{N}(0, \Op^{-p})$. Then
$G_p(\x,\y) = \mathbb{E}[u(\x)u(\y)]$.  The covariance function $G_p(
\cdot, \cdot )$, however, depends strongly on the boundary condition
of $\Op$, as can be seen in figure~\ref{fig:problem illustration}.

\subsection{Causes of boundary effects}
The reason for these boundary effects can be understood from either
PDE theory or from probability theory. To illustrate the PDE
perspective, consider the covariance operator $\Op^{-1}$ on $\Omega :=
[0,1]$ with homogeneous Dirichlet boundary conditions, and $\x \in
\Omega, \y \in \partial \Omega$. Then that $G(\x, \y) = 0$, since
$G(\x,\cdot)$ has to satisfy the boundary condition. By continuity,
Green's function is small near the boundary, even if $\y$ is only
close to the boundary. So for a Gaussian field $u \sim
\mathcal{N}(0,\covop)$ and $\x,\y\in \Omega$ near the boundary, $\cov(
u(\x), u(\y) )$ is smaller than what it would have been without the
boundary condition. To illustrate the probabilistic perspective,
consider the same operator $\Op$ and domain and $\gamma = 1$, only
with homogeneous Neumann boundary conditions. Loosely speaking, the
Green's function $G_1(\x,\y) = \left ( \Op^{-1}\delta_{\x}\right )
(\y)$ is the amount of time a particle spends near $\y$, given that it
started its walk at $\x$, if it is killed at rate $\kappa:=
\sqrt{\alpha / \gamma} = \sqrt{\alpha}$ \mycite{Oksendal03}. Then the
Green's function
\begin{equation*}
  G_2(\x,\y) = \left ( \Op^{-2}\delta_{\x}\right ) (\y) =
  \int_{\Omega} G_1(\x,\z)G_1(\z,\y)d\z
\end{equation*}
is interpreted as the amount of time a branching particle spends at
$\y$, had it started at $\x$ and if it is killed at the same rate
$\kappa$. The Neumann boundary means the particle reflects off the
boundary upon hitting it.\footnote{The probabilistic interpretation
  of Robin boundary conditions is involved---we refer to
  \mycite{SingerSchussOsipovEtAl08} for
  a numerical study.} Since the particle is reflected off the
boundary, it spends more time near it, making $G(\x,\y )$ large near
the boundary. Thus, the opposite happens --- the covariance is larger
near the boundary than what it would be without the boundary.  These
boundary effects ($G(\x, \cdot )$ is too big or too small near the
boundary) can be undesirable from a statistical modeling point of
view.  In the next section, we review approaches based on extending
the domain, and in sections \ref{section:robin} and
\ref{section:variance} we present two novel methods to mitigate these
boundary effects.

\section{Extending the domain}
The presented problem has a seemingly appealing solution---considering
an extended open domain $\Omega' \supset \Omega$ with sufficiently
regular boundary $\partial
\Omega'$, which is far enough from $\Omega$ that boundary
effects arising from $\partial\Omega'$ are negligible in $\Omega$.
In this section, we present variants of this approach and discuss challenges that arise
for large-scale problems.

Recall that we are particularly interested in scalable algorithms for
the application of the (discretized) covariance operator, its inverse
and its square root to vectors.  Before discussing concrete methods
that are based on domain extension, some comments are in
order.  First, extending the domain may result in undesired
correlations between parts of the domain. An extreme example would be
a domain which consists of two disjoint subdomains. In such a case, a
connected domain $\Omega'$ that encompasses these subdomains inevitably
introduces correlations between them. Second, creating an extended
domain $\Omega'$ comes at a cost, both in terms of development
time and computing time. For instance, it might require to extend a given
mesh for $\Omega$ to a mesh for the extended domain $\Omega'$, and
to manage the increased overall number of unknowns of the
problem. 

Let us start with introducing some notation. We consider the
covariance and the precision operators $\Op'^{-2}$ and $\Op'^2$,
respectively. Here, $\Op'$ is an elliptic PDE operator defined over
$\Omega'$, which incorporates, for instance, homogeneous Neumann or
Dirichlet boundary conditions at
$\partial\Omega'$. Assume we are given a discretization (e.g., based
on finite elements or finite differences) for functions defined over
$\Omega$, which we extend to a discretization of functions defined
over $\Omega'$.
We denote the number of degrees of freedom of the discretization for
functions defined over $\Omega'$ by $n$, and assume that the
corresponding unknowns are ordered such that the first $n_1$ unknowns
correspond to points that are inside or on the boundary of
$\Omega$. The remaining $n_2=n-n_1$ unknowns correspond to points in
$\Omega^c$, the domain extension.
This implies the following block structure of the covariance and 
precision matrices $\Sigma', Q'\in \mathbb R^{n\times n}$, respectively.
\begin{equation}\label{eq:block_matrices}
  \Sigma' =
  \begin{pmatrix} 
    \Sigma'_{11} & \Sigma'_{12} \\
    \Sigma'_{21} & \Sigma'_{22}
  \end{pmatrix},\qquad
    Q' =
  \begin{pmatrix} 
    Q'_{11} & Q'_{12} \\
    Q'_{21} & Q'_{22}
  \end{pmatrix}.
\end{equation}

In this setting, $\Sigma'_{11}\in \mathbb R^{n_1\times n_1}$ can be
used as covariance matrix for unknowns corresponding to points inside
$\Omega$. Note that the matrices $\Sigma',Q'$ in
\eqref{eq:block_matrices} might not be available in assembled form,
and we might only be able to apply them to vectors. The application of
the blocks to vectors can then be computed efficiently by appropriate
padding of vectors with zeros, followed by truncation. To be
precise, we denote by $P_1:\mathbb R^n\to \mathbb R^{n_1}$ and
$P_2:\mathbb R^n\to \mathbb R^{n_2}$ the (Boolean) operators that
restrict vectors to their first $n_1$ and to their last $n_2$
components, respectively. The corresponding adjoint operators
$P_1^*:\mathbb R^{n_1}\to \mathbb R^{n}$ and $P_2^*:\mathbb R^{n_2}\to
\mathbb R^{n}$ are padding-by-zero operators.  For instance, for $\bs
v \in \mathbb{R}^{n_1}$, we can efficiently compute $\Sigma'_{11}\bs
v$ as $ P_1\Sigma'P_1^*\bs v$. 

The precision for unknowns corresponding to points in $\Omega$ is
found as the Schur complement, \mycite{RoininenHuttunenJanneEtAl14}
\begin{align}
  \Sigma'^{-1}_{11} &= Q'_{11} -
  Q'_{12}Q'^{-1}_{22}Q'_{21}. \label{eq:schur}
\end{align}
Hence, fast application of the precision $\Sigma'^{-1}_{11}$ to
vectors requires that we can apply $Q'^{-1}_{22}$ efficiently. The
ability to do this depends on the specific choice of the
discretization, and we discuss some special cases
next. Additionally, we discuss options for applying the square root of
the covariance operator, $\Sigma'^{1/2}_{11}$, to vectors, as is
required for computing sample realizations from Gaussian distributions
with covariance matrix $\Sigma'_{11}$.

\subsection{Domain extension from \mycite{RoininenHuttunenJanneEtAl14}}
First, we summarize the approach proposed in
\mycite{RoininenHuttunenJanneEtAl14}, where the authors use finite
differences (for simple geometries) or finite elements (for more
complicated geometries) to discretize elliptic operators defined on
$\Omega'$. They assume that the matrices \eqref{eq:block_matrices} are
available in assembled form, which allows them to apply $Q'^{-1}_{22}$
using standard solvers for positive definite sparse matrices that are
available in assembled form. For computing samples, which requires a
square root of $\Sigma'_{11}$ or its inverse, they assume a
factorization $Q' = L^TL$, which is reasonable as we assumed that the
precision is the product of two elliptic PDE operators.  Computing
samples from the distribution is then carried out by defining
\begin{align*}
  \tilde{L}_{1} &= L_{1} - L_{2}Q'^{-1}_{22}Q'_{21}, 
\end{align*}
with $L_1= LP_1^*\in \mathbb{R}^{n \times n_1}$ are the first $n_1$
columns of $L$ and $L_2= LP_2^*\in \mathbb{R}^{n \times n_2}$ are the
last $n_2$ columns.
With a short calculation, one
verifies that
$\tilde{L}_1^T \tilde{L}_1 = \Sigma'^{-1}_{11}$. Thus, we can obtain
samples from $\mathcal{N}(0, \Sigma'_{11})$ by solving
$\tilde{L}_1^{T}\bs u = \bs z$, where $\bs z \sim \mathcal{N}(0, I_1)$,
with the $n_1\times n_1$-identity matrix $I_1$.
However, $\tilde{L}_1$ might not be sparse as $Q'^{-1}_{22}$ is in
general dense. Thus, having to solve a linear system with
coefficient matrix $\tilde{L}_1$ might not be feasible for large $n$.

\subsection{Modfication of approach from \mycite{RoininenHuttunenJanneEtAl14}}
A modification of the approach taken in
\mycite{RoininenHuttunenJanneEtAl14} is to start by defining how
samples from $\mathcal{N}(0, \Sigma'_{11})$ are generated.  Namely,
let $\bs u \sim P_1 L^{-1} \bs z$, with $\bs z$ a finite element
approximation to white noise. Then, a short argument yields that $\bs
u \sim \mathcal{N} (0,\Sigma'_{11})$ and using a decomposition
analogous to \eqref{eq:schur}, we may recover the corresponding
precision. Provided systems with $L$ can be solved
efficiently,\footnote{For finite element discretizations, proper
  definition of the covariance factor $L$ includes, besides an
  elliptic solve, a mass matrix square root
  \mycite{RoininenHuttunenJanneEtAl14, Bui-ThanhGhattasMartinEtAl13},
  which can make the efficient application of $L$ challenging.} this
provides a method to compute samples and to apply the covariance
matrix to vectors that does not require assembled matrices. However,
now the bottleneck is in the need to apply the precision which
requires the block $Q'^{-1}_{22}$ that is usually not available unless
one assembles the matrix $Q'$. One possibility is to use an iterative
method, such as the conjugate gradient method, to solve systems with
$Q'_{22}$. However, unless an efficient preconditioner for this solve
is available, this can require large numbers of iterations, as can be
seen in the following section.

\subsection{Domain extension and Fourier bases}
So far, we have considered approaches based on local discretizations
for the elliptic operator $\Op$. As an alternative, we may consider
using a global basis, such as a discretization based on the
(non-uniform) fast Fourier transform (FFT), which allows fast
application of $\Op'$ and $\Op'^{-1}$.  This requires the extended
domain ${\Omega'}$ to be a box, and enables fast application of
$\Sigma'$ and $Q'$ without requiring these matrices in assembled form.
However, similar as above, we do not have access to $Q'^{-1}_{22}$
because we cannot extract and invert the submatrix $Q'_{22}$. As
discussed above, one option would be to solve systems with $Q'_{22}$
iteratively. However, this requires a large number of iterations, making the
method inefficient in practice.

\subsection{Practical aspects}
Each of the approaches discussed above has limitations for large-scale
problems. Sampling using the method suggested in
\mycite{RoininenHuttunenJanneEtAl14} requires assembling and solving a
dense system. While our modification provides a fast method to compute
samples, it either requires matrix assembly or an iterative Krylov
method for applying the precision operator.  Using a global Fourier
basis on a rectangular domain extension requires an iterative method
for applying the precision operator as well. We found this to take a
large number of iterations that has to be performed each time the
precision is applied. The methods we proposed in the next section
require some upfront computation to estimate an optimal Robin
coefficient or the pointwise variance field. After this step, all
computations with the covariance operator can be performed
efficiently and without requiring assembled matrices.

\section{Optimal Robin boundary conditions}\label{section:robin}
In this section we aim at finding Robin boundary conditions that
mitigate the undesirable boundary effects shown in
figure~\ref{fig:problem illustration}. We derive a coefficient in the
Robin condition such that the Green's functions (which are
also the covariance functions) of the domain are close to the free
space Green's functions.

For $\x \in \Omega$, we denote the free-space Green's functions for
$\Op^p$, centered at $\x$, by $\Phi_{p}(\x,\cdot)$, $p=1,2$. Their
explicit expressions depend on the dimension $d$ of the domain,
see appendix~\ref{subsec:explicit}. For
$p>d/2$ these free-space Green's function are known as Mat\'ern
covariance functions.

The corresponding domain Green's functions with Robin boundary
condition are denoted by $G_{p}( \x, \cdot ), p =  1,2$, and they
satisfy:
\begin{subequations}
  \begin{alignat}{2}
    \Op G_1 
    &= \delta_{\x} 
    &&\quad \text{ in } \Omega,\\
    \beta G_1 + \frac{\partial G_{1} }{\dn} 
    &= 0 
    &&\quad \text{ on } \partial \Omega, \\
    \Op G_2 
    &= G_1 
    &&\quad \text{ in }\Omega, \\
    \beta G_2 + \frac{ \partial G_{2} }{ \dn } 
    &= 0 
    &&\quad \text{ on } \partial \Omega,
  \end{alignat}
\end{subequations}
where $\delta_\x$ is the Dirac-delta function centered at $\x\in
\Omega$, and $\beta:\partial \Omega\to \mathbb R_{\geq 0}$ is a
non-negative function defined for all boundary points $\y \in \partial
\Omega$. Following \mycite{Evans10}, we refer to the difference between
the free-space and the domain Green's functions,
\begin{equation}\label{eq:corrector}
  \phi^{\x}_{p} := \Phi_{p}( \x, \cdot ) - G_{p}( \x, \cdot )\quad\text{
    for }  p = 1,2,
\end{equation}
as the correctors. These correctors satisfy
the following equations:
\begin{subequations}\label{eq:corrector system}
  \begin{alignat}{2}
    \Op \phi^{\x}_1 
    &= 0  
    &&\qquad\text{ in } \Omega, \label{eq:corrector system1}\\
    \beta \phi^{\x}_1 + \frac{ \partial \phi^{\x}_1 }{ \dn } 
    &= \beta \Phi_1(\x,\cdot ) + \frac{ \partial \Phi_1(\x, \cdot) }{ \dn } 
    &&\qquad \text{ on } \partial \Omega, \label{eq:corrector system2}\\
    % % 
    \Op \phi^{\x}_2 
    &= \phi^{\x}_1  
    &&\qquad\text{ in } \Omega, \label{eq:corrector system3}\\
    \beta \phi^{\x}_2 + \frac{ \partial \phi^{\x}_2 }{ \dn } 
    &= \beta \Phi_2(\x, \cdot) + \frac{ \partial \Phi_2(\x, \cdot) }{ \dn } 
    &&\qquad\text{ on }  \partial \Omega. \label{eq:corrector system4}
  \end{alignat}
\end{subequations}
From \eqref{eq:corrector system}, it can be seen that if the right hand
sides in \eqref{eq:corrector system2} and \eqref{eq:corrector system4} were to
vanish everywhere on $\partial \Omega$, the correctors $\phi^{\x}_p \equiv 0$
and thus $\Phi_p(\x,\cdot) = G_p(\x,\cdot)$ for $p=1,2$.  If these
were to vanish for all $\x\in \Omega$, then $\Phi_{p} = G_{p}$ for
$p=1,2$. In the remainder of this section, we present
an optimization problem for the Robin coefficient $\beta(\y),
\y\in \partial \Omega$ that aim at making the boundary right hand
sides in \eqref{eq:corrector system2} and \eqref{eq:corrector system4} small,
and thus $\Phi_p \approx G_p$.

\subsection{One-dimensional case}\label{subsec:1D}
In one dimension, both $\Op$ and $\Op^2$ with appropriate boundary
condition are valid precision operators \mycite{Stuart10}. For $\Op$, we
only have to consider the system \eqref{eq:corrector system1} and
\eqref{eq:corrector system2}. The one-dimensional free-space Green's
function appearing in \eqref{eq:corrector system2}, is $\Phi_1(x,y) =
\frac{\exp( -\kappa |x - y| )}{2\kappa \gamma}$, with
 $\kappa = \sqrt{ \alpha / \gamma}$.  It can be verified that
for $\beta := \kappa$, the right hand side of 
\eqref{eq:corrector system2} vanishes. Thus, $G_1 = \Phi_1$, i.e.,
the domain and the free-space Green's functions coincide.

If one considers $\Op^2$ as covariance, this choice of $\beta$ does
\emph{not} guarantee that $G_2 = \Phi_2$. While for $\beta=\kappa$,
the right hand sides in \eqref{eq:corrector system2} and
\eqref{eq:corrector system3} vanish, the right hand side in the boundary condition
\eqref{eq:corrector system4} does not. Thus, for $\Op^2$ one should choose
a different value for $\beta$ following the approach presented in
section~\ref{subsec:beta} below.

\subsection{Higher dimensions}\label{subsec:beta}
We consider the precision operator $\Op^2$ and propose an optimization
approach for deriving an optimal Robin coefficient. As discussed
above, we would like to make $\beta ( \y )\Phi_p(\x,\y ) + \frac{
  \partial \Phi_p(\x, \y) }{ \dn }, p=1,2$ (the right hand sides of
\eqref{eq:corrector system2} and \eqref{eq:corrector system4})
vanish. For a fixed $\x \in \Omega$ and $\y\in \partial \Omega$, we
may choose, as a compromise, $\beta = \beta( \y )$ to be the average
of the roots of these terms. Note that both terms are linear in
$\beta$ with positive slopes $\Phi_p(\x,\y), p=1,2$. Recall that a
convex parabola attains its minimum value at the mean of its two
roots.  Thus, the parabola (in the variable $\beta$)
\begin{equation*}
  \left (\beta \Phi_1(\x,\y ) + \frac{ \partial \Phi_2(\x, \y) }{ \dn } \right )
  \left (\beta \Phi_2(\x,\y ) + \frac{ \partial \Phi_p(\x, \y) }{ \dn } \right )
\end{equation*}
attains its minimum in the average of its roots. Thus, for a fixed
$\x\in\Omega, \y \in \partial \Omega$, the minimum of the parabola may serve as
a compromise between the two competing terms. However, this compromise is made 
for a single $\x \in \Omega$. In order to take into account all
$\x \in \Omega$, we average, recovering the following optimization problem:
\begin{equation}\label{eq:beta-opt}
  \beta(\y ) := \argmin_{\beta \ge 0} 
  \frac{1}{|\Omega|}\int_{\Omega}
  \left(\beta \Phi_1(\x,\y) + \frac{\partial \Phi_1}{\dn}(\x,\y)
  \right)   \left(\beta \Phi_2(\x,\y) + \frac{\partial \Phi_2}{\dn}(\x,\y)
  \right) d\x.
\end{equation}
This quadratic minimization problem can be solved easily for $\beta$, noting that
the constraint $\beta \geq 0$ can be enforced on the solution. This leads to the
following expression for $\beta(\y)$:

\begin{equation}\label{eq:beta}
  \begin{split}
  \beta(\y) &:= \max(0,\tilde\beta(\y)),\\
    \tilde\beta(\y) 
    &= - \frac{\int_{\Omega}  \Phi_1(\x,\y) \frac{\partial \Phi_2}{\dn}(\x,\y) +
      \Phi_2(\x,\y) \frac{\partial \Phi_1}{\dn}(\x,\y) d\x}{2\int_{\Omega}
      \Phi_1(\x,\y) \Phi_2(\x,\y) d\x}.
  \end{split}
\end{equation}
Note that the integrals occurring in \eqref{eq:beta} are finite for
all dimensions $d=1,2,3$. Computing $\tilde\beta(\y)$ requires the
computation of two integrals over $\Omega$.  From the explicit
expressions \eqref{eq:explicit beta 2D}
and \eqref{eq:explicit beta 3D} for $\tilde\beta(\y)$, one can verify
that $\tilde\beta(\y) > 0$, if $\Omega$ is convex.

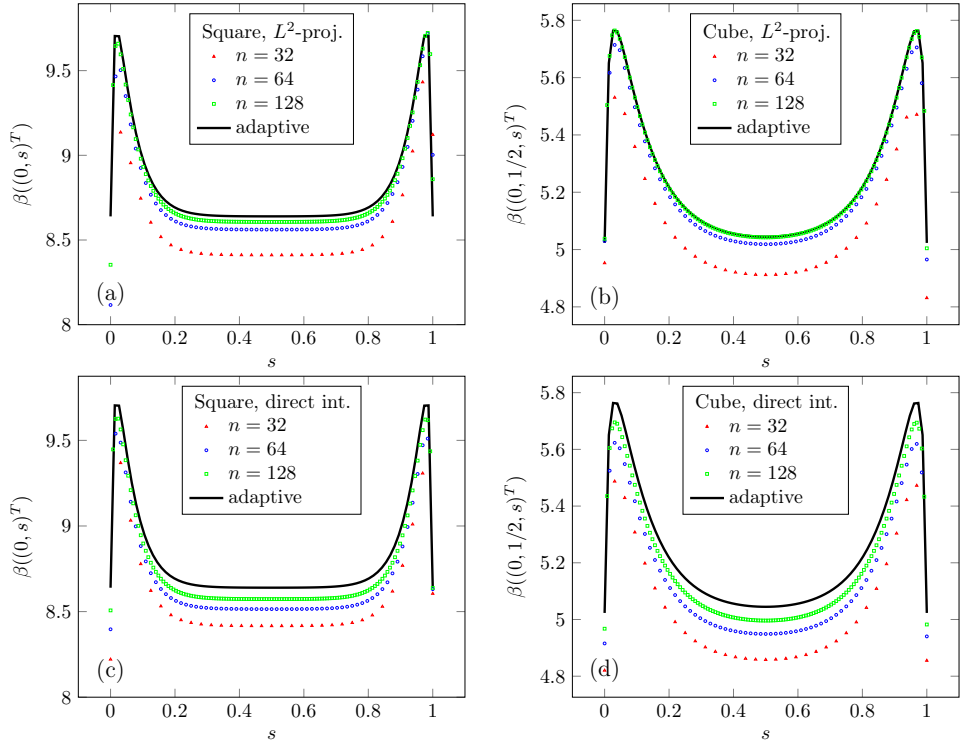
\begin{figure}
  \begin{tikzpicture}[thick,scale=0.75, every node/.style={scale=1}]
    \begin{axis}
      [
      compat = 1.3,
      ymin = 8,
      xlabel = {$s$},
      ylabel = {$\beta((0,s)^T)$},
      legend style={at={(0.5,0.97)},anchor=north},
      legend cell align=left,
      ]
      \addlegendimage{empty legend}
      \addlegendentry{\hspace{-.6cm}{Square, $L^2$-proj.}}
      \addplot
      [mark=triangle, red, only marks, mark size=.7]
      table {square_beta_radial_32.txt};
      \addplot
      [mark=o, blue, only marks, mark size=.7] table
      {square_beta_radial_64.txt};
      \addplot
      [mark=square, green, only marks, mark size=.7]
      table {square_beta_radial_128.txt};
      \addplot
      [very thick, mark=none, black] 
      table {square_beta_adaptive_2.txt};
      \node at (0, 50) {\Large (a)};
      \addlegendentry{$n=32$}
      \addlegendentry{$n=64$}
      \addlegendentry{$n=128$}
      \addlegendentry{adaptive}
    \end{axis}
  \end{tikzpicture}\hfill
  \begin{tikzpicture}[thick,scale=0.75, every node/.style={scale=1}]
    \begin{axis}
      [
      compat = 1.3,
      xlabel = {$s$},
      ylabel = {$\beta((0,1/2,s)^T)$},
      legend style={at={(0.5,0.97)},anchor=north},
      legend cell align=left
      ]
      \addlegendimage{empty legend}
      \addlegendentry{\hspace{-.6cm}{Cube, $L^2$-proj.}}
      \addplot
      [red, mark=triangle, only marks, mark size = .7]
      table {cube_beta_radial_32.txt};
      \addplot
      [blue, mark=o, only marks, mark size=.7] 
      table {cube_beta_radial_64.txt};
      \addplot
      [green,  mark=square, only marks, mark size=.7]
      table {cube_beta_radial_128.txt};
      \addplot
      [very thick, black, mark=none]
      table {cube_beta_adaptive_2.txt};
      \node at (0, 0) {\Large (b)}; 
      \addlegendentry{$n=32$}
      \addlegendentry{$n=64$}
      \addlegendentry{$n=128$}
      \addlegendentry{adaptive}
    \end{axis}
  \end{tikzpicture}\\
 \begin{tikzpicture}[thick,scale=0.75, every node/.style={scale=1}]
    \begin{axis}
      [
      compat = 1.3,
      ymin = 8,
      xlabel = {$s$},
      ylabel = {$\beta((0,s)^T)$},
      legend style={at={(0.5,0.97)},anchor=north},
      legend cell align=left
      ]
      \addlegendimage{empty legend},
      \addlegendentry{\hspace{-.6cm}{Square, direct int.}}
      \addplot
      [mark=triangle, red, only marks, mark size=.7]
      table {square_beta_std_32.txt};
      \addplot
      [mark=o, blue, only marks, mark size=.7] table
      {square_beta_std_64.txt};
      \addplot
      [mark=square, green, only marks, mark size=.7]
      table {square_beta_std_128.txt};
      \addplot
      [very thick, mark=none, black] 
      table {square_beta_adaptive_2.txt};
      \node at (0, 15) {\Large (c)}; 
      \addlegendentry{$n=32$}
      \addlegendentry{$n=64$}
      \addlegendentry{$n=128$}
      \addlegendentry{adaptive}
    \end{axis}
  \end{tikzpicture}\hfill
  \begin{tikzpicture}[thick,scale=0.75, every node/.style={scale=1}]
    \begin{axis}
      [
      compat = 1.3,
      xlabel = {$s$},
      ylabel = {$\beta((0,1/2,s)^T)$},
      legend cell align=left,
      legend style={at={(0.5,0.97)},anchor=north},
      ]
      \addlegendimage{empty legend},
      \addlegendentry{\hspace{-.6cm}{Cube, direct int.}}
      \addplot
      [red, mark=triangle, only marks, mark size = .7]
      table {cube_beta_std_32.txt};
      \addplot
      [blue, mark=o, only marks, mark size=.7] 
      table {cube_beta_std_64.txt};
      \addplot
      [green,  mark=square, only marks, mark size=.7]
      table {cube_beta_std_128.txt};
      \addplot
      [very thick, black, mark=none]
      table {cube_beta_adaptive_2.txt};
      \node at (0, 0) {\Large (d)}; 
      \addlegendentry{$n=32$}
      \addlegendentry{$n=64$}
      \addlegendentry{$n=128$}
      \addlegendentry{adaptive}
    \end{axis}
  \end{tikzpicture}
  \caption{Optimal Robin boundary coefficients $\beta$ for an edge of
    a square using $\Op = -\Delta + 121$ (a), (c) and a line on
    a face of a cube using $\Op = -\Delta + 25$ (b), (d). Shown
    are coefficients computed by adaptive quadrature, and their
    discrete approximations on regular meshes obtained by dividing
    $n^2$ squares into $4n^2$ triangles in two dimensions, and $n^3$
    cubes into $6n^3$ tetrahedra in three dimensions. The
    approximations are either based on approximate $L_2$-projections
    followed by finite element quadrature (a), (b) or on direct finite
    element quadrature (c), (d) as discussed in section
    \ref{subsec:practial_robin}.\label{fig:beta}}
\end{figure}

\subsection{Practical aspects}\label{subsec:practial_robin}
Numerical evaluation of the singular integrals in \eqref{eq:beta} is a
challenging task. We have used two practical approaches for computing
Robin coefficients in the context of finite element discretizations.

The first approach approximates the fundamental solutions with
piecewise constants, found by evaluating the fundamental solutions and
their derivatives at element centers. This avoids singularities and
the integrals in \eqref{eq:beta} for the resulting piecewise constant
functions can be computed exactly. Robin boundary coefficients
computed using this approach are shown in (c) and (d) in figure
\ref{fig:beta}. As the mesh is refined, the Robin coefficients
converge to the numerically accurate Robin coefficient, which is
obtained from adaptive quadrature \mycite{Cubature}.

Our second approach is motivated by the derivation of $\beta$ as
presented in section \ref{section:robin}, but for the
\emph{discretized} problem. We consider discrete approximations
$\Phi_1^h$ and $\Phi_2^h$ of the free-space Green's functions $\Phi_1$
and $\Phi_2$, and aim at solving the optimization problem
\eqref{eq:beta-opt} with these discrete Green's functions rather than
their continuous counterparts. Our motivation is that $\Phi_1$ and
$\Phi_2$ cannot be represented in finite dimensions and thus
the discrete domain Green's functions can never be good approximations
to the continuous free-space Green's functions. The best we can hope
for is that the numerically computed domain Green's functions
approximate discrete free-space Green's functions $\Phi_1^h$ and
$\Phi_2^h$.  Unless for uniform discretizations, $\Phi^h_1(\x,\cdot)$
and $\Phi^h_2(\x,\cdot)$ depend on the discretization mesh in a
neighborhood of $\x$ and thus would have to be computed for every
$\x$. To avoid this, and using the radial symmetry of Green's
functions, we compute a one-dimensional finite element approximation
to the free-space Green's function as illustrated next for $d=2$---an
analogous approach can be taken for $d=3$. Recall that for a radially
symmetric function $v$, we can use polar coordinates $(r,\phi)$ for the Laplacian
operator to write:
\begin{align*}
  \Delta v &= \frac{\partial^2 v}{\partial r^2} + \frac{1}{r}\frac{\partial v}{\partial r}.
\end{align*}
Hence, using a Dirac-delta $\delta_0$, we find the weak form:
\begin{align*}
  v(0) &= \int_{\mathbb{R}^2} v\delta_{0} \,d\x \\
  &= \int_{\mathbb{R}^2} v(-\gamma \Delta + \alpha)\Phi_1 \, d\x \\ 
  &= \int_{0}^{2\pi} \!\!\int_{0}^{\infty} 
  (-\gamma\frac{\partial^2 \Phi_1}{\partial r^2} -
  \gamma\frac{1}{r}\frac{\partial \Phi_1}{\partial r} + \alpha \Phi_1) v r \,dr d\theta \\
  &= 2\pi \int_{0}^{\infty} (\gamma \frac{\partial \Phi_1}{\partial r}
  \frac{\partial v}{\partial r} + \alpha \Phi_1v) r \,dr,
\end{align*}
where the last equality follows from integration by parts and radial
symmetry.  Now, we solve for $\Phi_1$ as a function of the radius $r$
using the finite element method in one dimension. The space
discretization length scale $h$ and the polynomial order for this
one-dimensional finite element calculation should be representative of
their higher-dimensional counterparts, such that the resulting
discrete free-space Green's functions can be used to compute the
optimal Robin coefficient functions for the discrete problem. We
truncate the integration over the radius to $[0,R]$, with $R$
sufficiently large such that the Neumann boundary condition imposed at
$r=R$ has negligible effect. To compute $\Phi_2$ as a function of $r$,
we solve the discretized problem twice.

The usual finite element quadrature can now be used for computing the
Robin coefficients, since the numerically computed free-space Green's functions are finite
element functions (or interpolations of radially symmetric
one-dimensional finite element functions to a two/three-dimensional
mesh). The results are shown in (a) and (b) in figure \ref{fig:beta}.
Moreover, Robin coefficients computed with these discrete free-space
Green's functions are (close-to) optimal for a discrete version of the
optimization problem \eqref{eq:beta-opt}, which is particularly
relevant for coarser discretizations, i.e., in the pre-asymptotic
regime.

\section{Normalizing the variance}\label{section:variance}
The approach presented in this section can be used to mitigate
boundary effects in covariance operators derived from elliptic PDEs
with Neumann or Robin boundary conditions. In particular, it can be
used in combination with the optimal Robin coefficient approach
introduced in the previous section.
Recall, that the correlation between two (real valued)
random variables $X,Y$ is defined as
\begin{align*}
  \corr(X,Y) := \frac{ \cov(X,Y)}{\sqrt{ \var(X) \var(Y) } }.
\end{align*}
Now, let us consider a Gaussian random field, $u$, which is defined
over $\Omega$ and has the covariance function $G_2$ with Robin or
Neumann boundary conditions, and a Gaussian random field, $v$, defined
over $\mathbb{R}^d$ with covariance function $\Phi_2$. Then,
\begin{alignat*}{2}
  \corr(u(\x) ,u(\y)) 
  &= \frac{ G_2(\x,\y) }{\sqrt{ G_2(\x,\x) G_2(\y,\y) } } 
  &&\qquad \text{ for } \x,\y \in \Omega,\\
  \corr(v(\x) ,v(\y)) 
  &= \frac{ \Phi_2(\x,\y) }{\sqrt{ \Phi_2(\x,\x) \Phi_2(\y,\y) } }
  &&\qquad \text{ for } \x,\y \in \mathbb{R}^d.\\
\end{alignat*}
A key property of $v$ is that
\begin{align}\label{eq:isotropic property}
  \Phi_2(\x,\x) = \cov(v(\x),v(\x)) = \var( v(\x) )  = \sigma^2\: \forall \x \in \mathbb{R}^d,
\end{align}
where $\sigma^2$ is a constant given explicitly in~\eqref{eq:sig2}.  
This means that the covariance of the field $v$ coincides with its correlation
(up to a multiplicative constant). This is a desirable property 
from a modeling point of view, as it means that
$v(\x)$ and $v(\y)$ vary on the same scale. Said differently, it is as 
likely to observe $v(\x)$ at a certain distance from its mean $\mathbb{E}[v(\x)]$
as it is to observe $v(\y)$ at the same distance from its mean $\mathbb{E}[v(\y)]$.
This is not the case, however, for $u$. A property similar to~\eqref{eq:isotropic property}
does not hold for $\var( u(\x ) )= G_2(\x,\x )$. This, as will be seen in
the numerical simulations, is a significant
part of the boundary effect illustrated in figure~\ref{fig:problem illustration}. 
The idea of the approach proposed in this section is to modify the covariance
operator $\Op^{-2}$ so that its Green's functions satisfy
\eqref{eq:isotropic property} (with the constant $\sigma^2$).%

Before presenting our method in function space, we consider its
simpler analogue in $\mathbb{R}^n$.  Consider a (symmetric positive
definite) covariance matrix $\Sigma \in \mathbb{R}^{n \times n}$ with
non-constant diagonal and define a diagonal matrix $D$ by $D_{ii} =
\sigma^{-1}\Sigma_{ii}$, with $\sigma>0$.  Let $\Lambda :=
D^{-\frac12}\Sigma D^{-\frac12}$ and $v \sim \mathcal{N}( 0, \Lambda
)$.  Then, \eqref{eq:isotropic property} holds for $v$ in the sense
that
\begin{align*}
   \Lambda_{ii} = \cov( v_i,v_i) = \var( v_i ) = \sigma^2, \text{ for } 1 \leq i \leq n.
\end{align*}
The covariance operator modification presented below is the
infinite-dimensional analogue to the computation of $\Lambda$.

Consider $\Op$ as in section \ref{section:preliminaries}, equipped
with a homogeneous Robin boundary condition $\beta u +\frac{\partial
  u}{\dn} = 0$ with $\beta:\partial \Omega \to \mathbb{R}_{\ge 0}$
bounded. Note that this includes a homogeneous Neumann boundary
condition for $\beta\equiv 0$.  We define $g(\x) :=
{\sigma}/\sqrt{G_2( \x, \x )}$, the infinite-dimensional analogue of
the matrix $D^{-\frac12}$ defined above. Note that $G_2( \x, \x )$ is
the pointwise variance field of $\mathcal{N}(0, \Op^{-2})$ and
$\sigma^2$ is the pointwise variance of the free-space covariance
function defined in \eqref{eq:sig2} in the appendix.  In the appendix
(proposition \ref{prop:g}) we show that $g$ is bounded away from zero
and infinity and that it is differentiable. Thus, $\covop :=
g\Op^{-2}g$ is a valid covariance operator and has constant pointwise
variance $\sigma^2$ (proposition \ref{prop:covar}).  Also, $u \sim
\mathcal{N}(0,\covop)$ are characterized by $u \sim gv$, where $v \sim
\mathcal{N}(0, \Op ^{-2} )$, which is a consequence of
\mycite[Proposition 1.18]{Prato06}.

Note that this transformation can be interpreted probabilistically
using particles that follow a Brownian motion.  For simplicity, we
assume $d=1$ such that $\Op^{-1}$ with Neumann boundary conditions is
a valid covariance operator. Then, the time a particle starting at
$\x\in \Omega$ spends in a set $A\subset\Omega$ before being killed
(killing occurs independently at a rate $\kappa^2=\alpha / \gamma$) is $\int_A
G_1(\x,\y)\,d\y$.  Multiplying $\Op$ by $g$ changes both the Laplacian
part of the operator (responsible for Brownian motion) and the
$\kappa^2$ (responsible for killing of particles).  Multiplying the
Laplacian by $g$ corresponds to a time change. This does not change
the distribution of Brownian paths (without killing), but changes the
particle velocities along the paths.  If one only changes the
traveling speed, one changes the measure on paths, because the rate of
killing stays the same. If the killing rate is changed by the same
factor, one obtains the same distribution of paths but particles are
sped up where the pointwise variance was too large and slowed down
where it was too small.

\subsection{Practical aspects}\label{subsec:practical variance}
Note that this method requires knowledge of the pointwise variance 
of the covariance operator $\Op^{-2}$ with some choice of boundary
conditions. This can be an expensive computation, but there are
several options to approximate the pointwise variance field.

One option is to calculate the pointwise variance through samples. For
the finite element method, this involves applying a square root of the
mass matrix $M$ to vectors
\mycite{Bui-ThanhGhattasMartinEtAl13}. Since this can be a difficult
task, we suggest an alternative. Denote by $K$ the symmetric finite
element discretization of $\Op$.  Then, the covariance matrix is
$K^{-1}MK^{-1}$ \mycite{Bui-ThanhGhattasMartinEtAl13}, and pointwise
variances of the finite element function are known to be the diagonal
entries of the covariance matrix. If we set $Z \sim \mathcal{N}(0,I)$
and let $X= K^{-1}Z, Y = K^{-1}MZ$ we get that
  \begin{align*}
    \cov(X,Y) = \mathbb{E}[XY^T] =  \mathbb{E}[K^{-1}ZZ^TM^TK^{-T}] = K^{-1}MK^{-1}.
  \end{align*}
Thus, we may estimate the pointwise variance as follows. Draw $\{Z_k\}_{k=1}^N$
iid as above, set $X_k= K^{-1}Z_k, Y_k = K^{-1}MZ_k$. Then the
pointwise variance is $\frac 1N \Sigma_{k=1}^N X_k \circ Y_k$, where
$(v\circ u)_i = v_iu_i$ (Hadamard product).

Additionally, often symmetry properties of the domain $\Omega$ can be used to speed up
the computation of the pointwise variance (as, e.g., in \mycite{Bui-ThanhGhattasMartinEtAl13}),
or an approximation for the pointwise variance
field, which is typically smooth, can be obtained through
interpolation with a small number of points.

The problem of estimating the diagonal of a matrix inverse has been
studied extensively in the literature. For fast estimation methods for
diagonals of Green's functions we refer to \mycite{RueMartino07,
  LinLuYingEtAl09}. Alternatively, low-rank matrix approximation of
the discretized covariance operator can be used to approximate the
diagonal.
The problem is considered for a wider class of matrices in
\mycite{BekasCurioniFedulova09, BekasKokiopoulouSaad07} using
stochastic estimation. A method based on applying an inverse
of a sparse matrix to carefully chosen vectors is proposed in \mycite{TangSaad12}. 

\section{Numerical Experiments}\label{section:numerics}
In this section, we study the ability of the methods proposed in
sections \ref{section:robin} and \ref{section:variance} to mitigate
boundary effects in two and three-dimensional numerical examples.
For comparison, we also present results obtained with homogeneous
Neumann boundary conditions as used in
\mycite{Bui-ThanhGhattasMartinEtAl13,IsaacPetraStadlerEtAl15}, with
homogeneous Dirichlet conditions, and with
the constant Robin coefficient as suggested in
\mycite{RoininenHuttunenJanneEtAl14}.
We use the finite element library FEniCS \mycite{LoggMardalWells12}
for our numerical tests,\footnote{The code to reproduce our results may
  be downloaded from \url{https://github.com/yairdaon/covariances}.}
and rely on linear finite elements in our computations.
Unless otherwise specified, we compute the Robin boundary coefficient
using the numerically computed approximate $L^2$-projection of Green's
functions discussed in section \ref{subsec:practial_robin}. For the 
parallelogram and Antarctica meshes we calculate the pointwise variance
at every discretization point $\x\in \Omega$ directly as $(\Op^{-2}
\delta_{\x} )(\x)$. For the cube mesh
we do so using our stochastic estimator derived in section 
\ref{subsec:practical variance} with $10,000$ samples,
which we find to result in reasonable approximations.

\begin{figure}
  \begin{tikzpicture}[thick,scale=1, every node/.style={scale=1}]
    \begin{axis}
      [
      % yscale = 0.75,
      % xscale = 0.75,
      compat = 1.3,
      xmin = 0.004,
      xmax = 0.5,
      xlabel = {$s$},
      ylabel = {$c(\x^{\star}, \x(s) )$},
      ymin   = 0,
      ymax   = 0.0008,
      yticklabels = , 
      ytick = \empty,
      legend style={nodes=right},
      legend style={at={(1,1)},anchor=north east},
      legend cell align=left
      ]
      \draw[black!30!white, thin] (25,0) -- (25,130);
      \addplot 
      [thick, black, mark=none]
      table
      {parallelogram_Free_Space_Greens_Function.txt};
      \addlegendentry{Free-Space};
      \addplot 
      [thick, red, dotted, mark=none] 
      table 
      {parallelogram_Neumann_Greens_Function.txt};
      \addlegendentry{Neumann BC};
      \addplot 
      [thick, green!80!black, loosely dashdotted, mark=none] 
      table 
      {parallelogram_Roininen_Robin_Greens_Function.txt};
      \addlegendentry{Constant Robin
        \mycite{RoininenHuttunenJanneEtAl14}};
      \addplot 
      [thick, green!60!black, densely dashdotted       ,   mark=none] 
      table 
      {parallelogram_Ours_Greens_Function_Radial.txt};
      \addlegendentry{Var.\ Robin (Sec.~\ref{subsec:beta})};
      
      \addplot 
      [red, thick, dashed, mark=none] 
      table 
      {parallelogram_Ours_Constant_Variance_Greens_Function_Radial.txt};
      \addlegendentry{Var.\ Robin+Const.\ Var.};
      
      \addplot 
      [thick, blue, dashed, mark=none] 
      table 
      {parallelogram_Neumann_Constant_Variance_Greens_Function.txt};
      \addlegendentry{Neumann+Const.\ Var.};
    \end{axis}
  \end{tikzpicture}
   \begin{tikzpicture}[thick, scale = 1, every node/.style={scale=1}]
     \clip (0,-0.85) rectangle (3.2,10);
     \begin{axis}
      [
      % yscale = 0.75,
      % xscale = 0.75,
      xmin = -0.01,
      xmax = .5,
      ymin   = -0.01,
      ymax   = .5,
      xlabel = {$$},
      legend cell align=left
      ]      
      \addplot [mark=none, fill=black!10!white] table {parallelogram_vertices.txt}; 
      \addplot [thin, black, mark=none] table {parallelogram_vertices.txt};
      \addplot [thick, red  , mark=none] table {parallelogram_line.txt};
      \addplot [only marks, mark = *   ] table {parallelogram_source.txt};
      \node at (2.4cm,3.cm) {\large $\Omega$};
      \node at (1.1cm,1.1cm) {$\x^\star$};
      \draw [black, thin] (0.9cm,0.9cm) -- (0.42cm,0.35cm);
      \node at (1.8cm,2.1cm) {cross section};
      \draw [black, thin] (1.6cm,1.95cm) -- (1.6cm,0.95cm);
      \draw[black] (23.75,0) -- (23.75,130);
    \end{axis}
  \end{tikzpicture}
  \caption{The left plot shows covariance functions derived from PDE
    operators with different
    boundary conditions for the
    parallelogram domain example (section \ref{subsec:parallelogram}).
    Shown are slices of the Green's function along a cross section.
    The right plot shows part of the parallelogram domain $\Omega$.
    The black dot is
    $\x^{\star} = (0.025, 0.025)^T$---the center of the
    Green's functions. 
    The red line indicates the cross section
    $\x(s) = (s, 0.6s + 0.01 )$, which is used in the left plot.
  \label{fig:parallelogram greens}}

\end{figure}
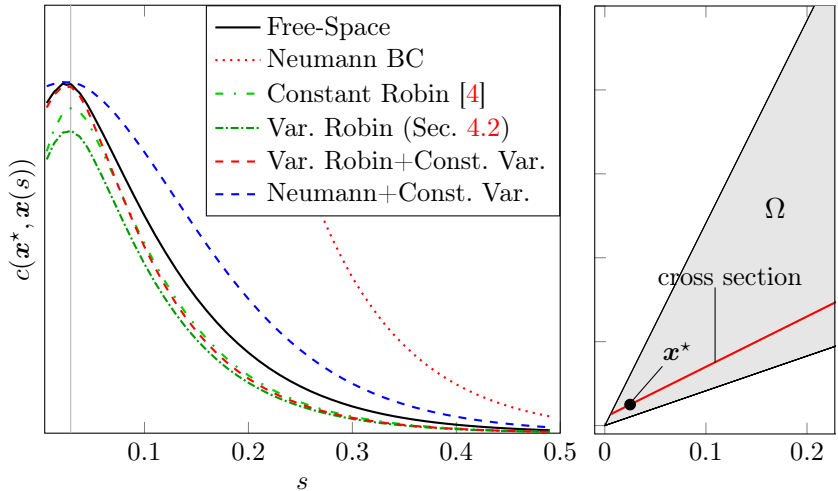

\subsection{Parallelogram example}\label{subsec:parallelogram}
We first illustrate our methods on a two-dimensional domain that is
more challenging than the square domain in figure~\ref{fig:problem
  illustration}. The results shown in figure~\ref{fig:parallelogram
  greens} show cross sections through covariance functions centered at
a particularly challenging point close to a corner of the domain. We
use $\Op = -\Delta + 121$ as the square root of the precision
operator. We discretized the unit square by $128^2$ points and then
transformed it to the parallelogram using a linear transformation such
that its vertices become $(0,0)$, $(\cos \theta_{-}, \sin
\theta_{-})$, $(\cos \theta_{+}, \sin \theta_{+})$ and
$(\cos\theta_{-} + \cos\theta_{+}, \sin\theta_{-} + \sin\theta_{+})$,
where $\theta_{\pm} = \frac{\pi}{4} \pm \frac{\pi}{8}$.  As can be
seen in figure~\ref{fig:parallelogram greens}, using Robin boundary
conditions results in a significant improvement over Neumann boundary
conditions. In this problem, the constant Robin coefficient
$\beta=\sqrt{\alpha}/1.42$ from \mycite{RoininenHuttunenJanneEtAl14}
and the variable Robin coefficient perform similarly. Moreover, figure
\ref{fig:parallelogram greens} also shows that the variance
normalization method results in constant pointwise variances, but that
the resulting covariance functions differ from the free-space
covariance functions. Combining varying Robin boundary conditions with
variance normalization leads to the best results.

\subsection{Antarctica domain example}\label{subsec:Antarctica}
We also show Green's functions on the Antarctica domain used for
Bayesian inference in \mycite{IsaacPetraStadlerEtAl15}. We use 
$\Op = -\Delta + \alpha, \alpha = 10^{-5}$ as the square root of the
precision operator, and measure distances in kilometres
(Antarctica extends laterally between 3000 and 6000 kilometers). 
Note that in \mycite{IsaacPetraStadlerEtAl15}, the authors used
$\alpha = 10^{-6}$, 
which leads to stronger point correlation.\footnote{To be precise, in
  \mycite{IsaacPetraStadlerEtAl15}, the authors used $\Op=10(-\Delta +
  10^{-6})$.}  
We used a finite element discretization with 27,749 cells.
Figure~\ref{fig:antarctica greens} shows a comparison of two domain
Green's functions, one centered far and one close to the
boundary. The differences between the covariance functions on the left
of the domain (which is West Antarctica) is due to the different
boundary conditions. As for the previous example, we find that using
Robin boundary conditions largely mitigates undesired boundary
effects. 

\begin{figure}%[htb]
  \begin{tikzpicture}
    \node at (-4.5,1.25) {
      \includegraphics 
      [width=0.45\textwidth]
      {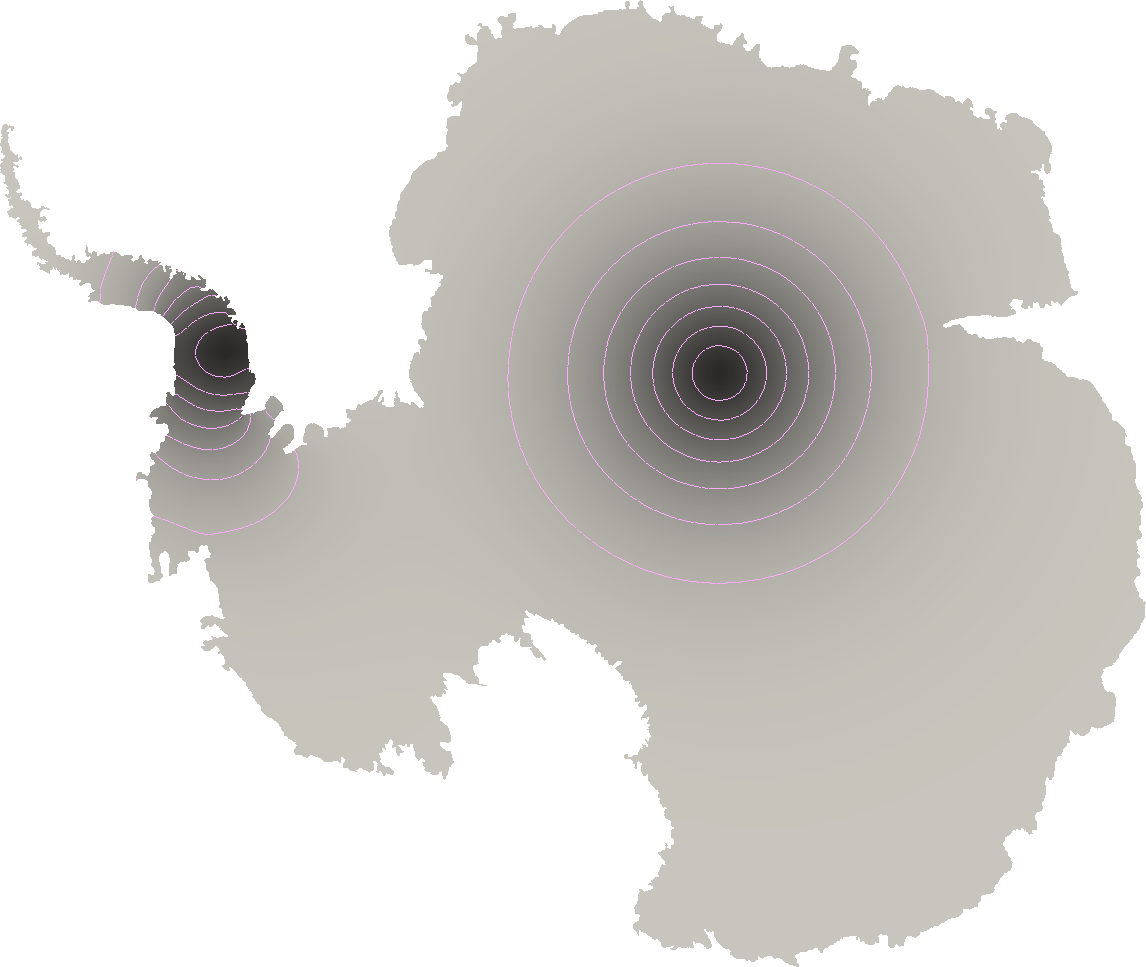}
    };
    \node at (-6,     3) {\large (a)};  

    \node at (-0.2,2.25) {
      \includegraphics
      [width=0.16\textwidth, trim=0 15cm 28cm 4cm, clip]
      {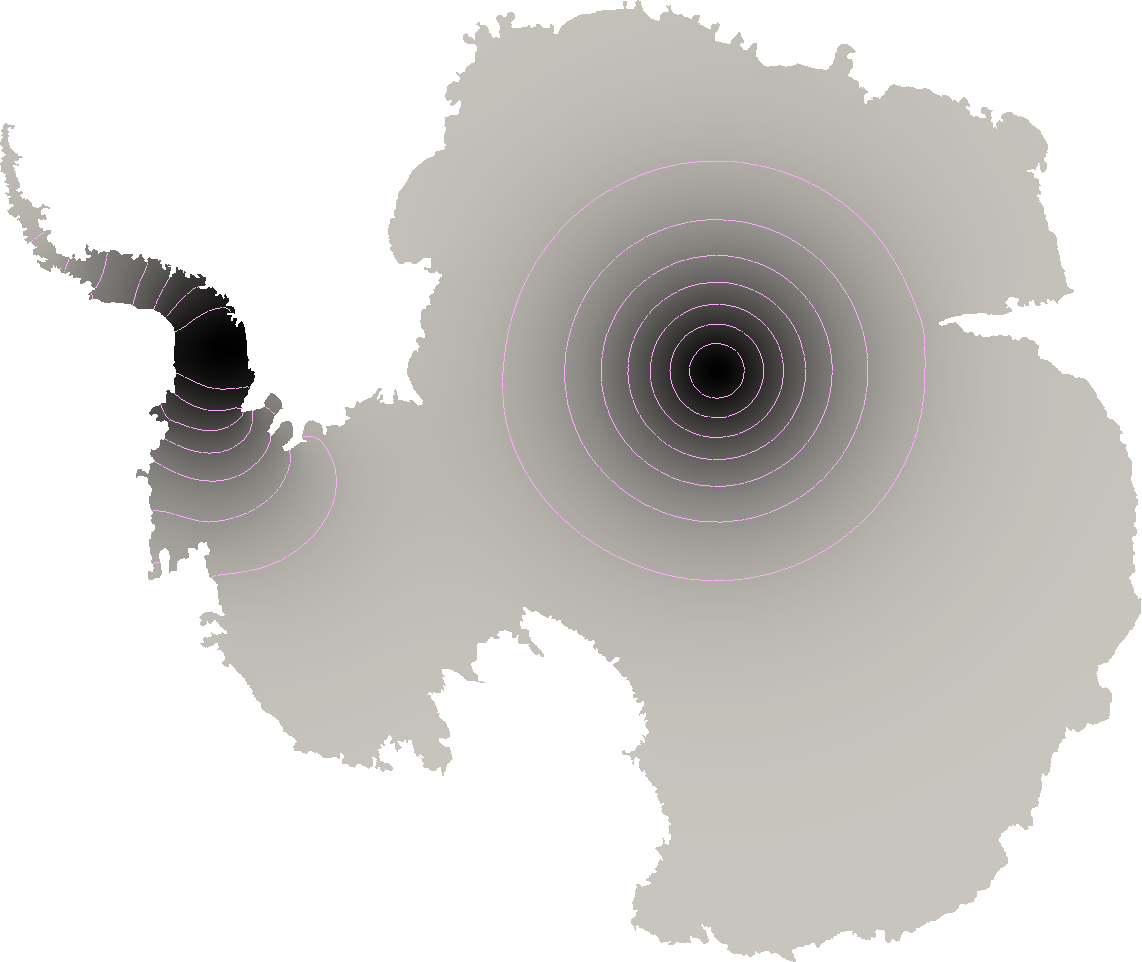}
    };
    \node at (0,      3) {\large (b)};

    \node at (2,2.25) {
      \includegraphics 
      [width=0.16\textwidth,  trim=0 15cm 28cm 4cm, clip]
      {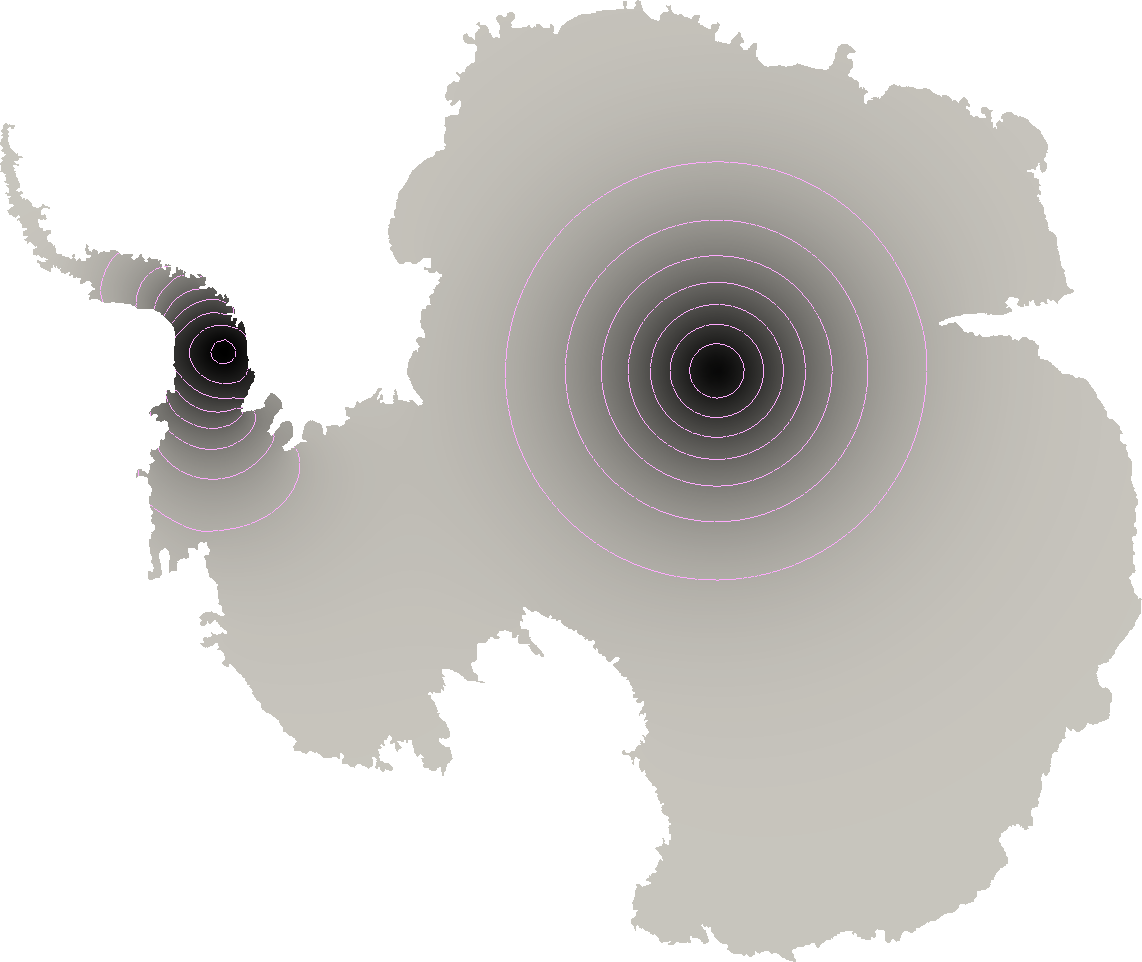}
    }; 
    \node at (2.25,   3) {\large (c)};

    \node at (-0.2,0) {
      \includegraphics
      [width=0.16\textwidth, trim=0 15cm 28cm 4cm, clip]
      {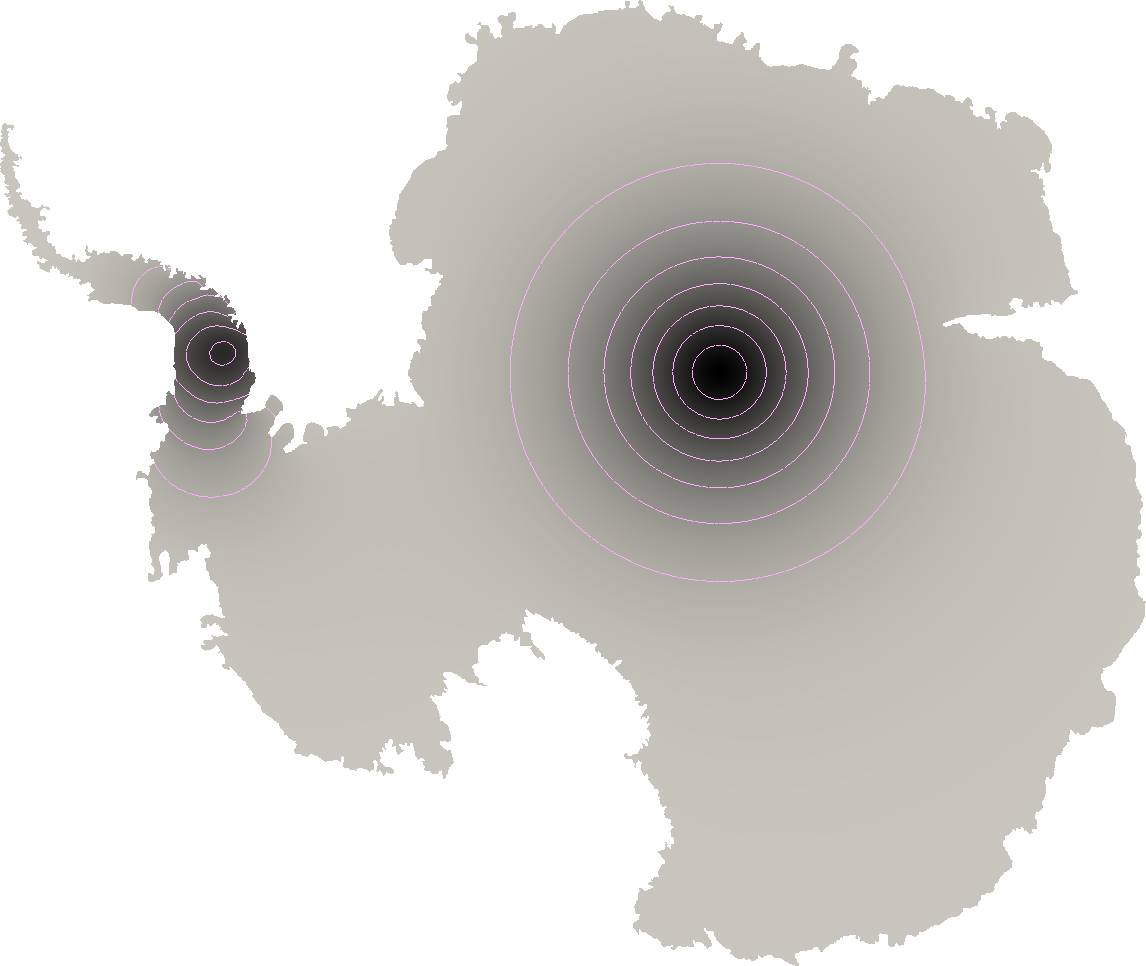}
    };
    \node at (0,    0.7) {\large (d)};

    \node at (2,0) {
      \includegraphics
      [width=0.16\textwidth,  trim=0 15cm 28cm 4cm, clip]
      {antarctica_Neumann_Greens_Function.png}
    }; 
    \node at (2.25, 0.7) {\large (e)};

    \node (colorbar) at (-8,1) {
      \includegraphics[width=0.11\textwidth]
      {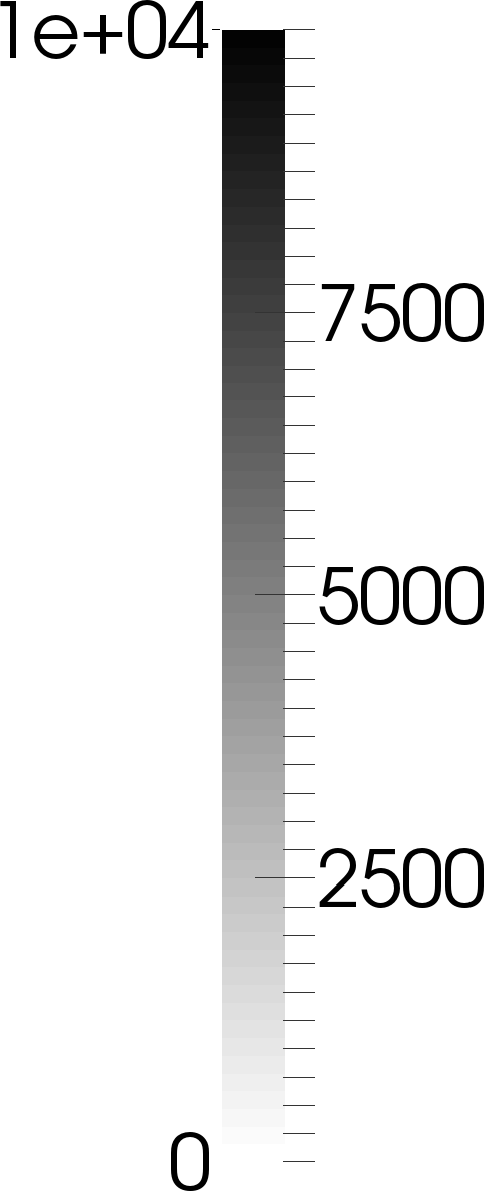}
    };
    \node at (-8,3.1) {%\large
    \begin{tabular}{c}
      co-\\ variance 
    \end{tabular}
  };
  \end{tikzpicture}
  \caption{Green's functions for the Antarctica
    domain detailed in section \ref{subsec:Antarctica}.
    Results for optimal Robin boundary conditions
    combined with variance normalization are shown
    in (a). These results should be compared
    with figure~\ref{fig:problem illustration}, which uses homogeneous
    Neumann boundary conditions. Magnifications are shown for 
    Neumann conditions with normalized variance (b),
    varying Robin boundary condition from section \ref{section:robin} (c), 
    Robin condition with constant coefficient taken from \mycite{RoininenHuttunenJanneEtAl14} (d),
    and Neumann boundary condition (e).}
  \label{fig:antarctica greens}
\end{figure}
 
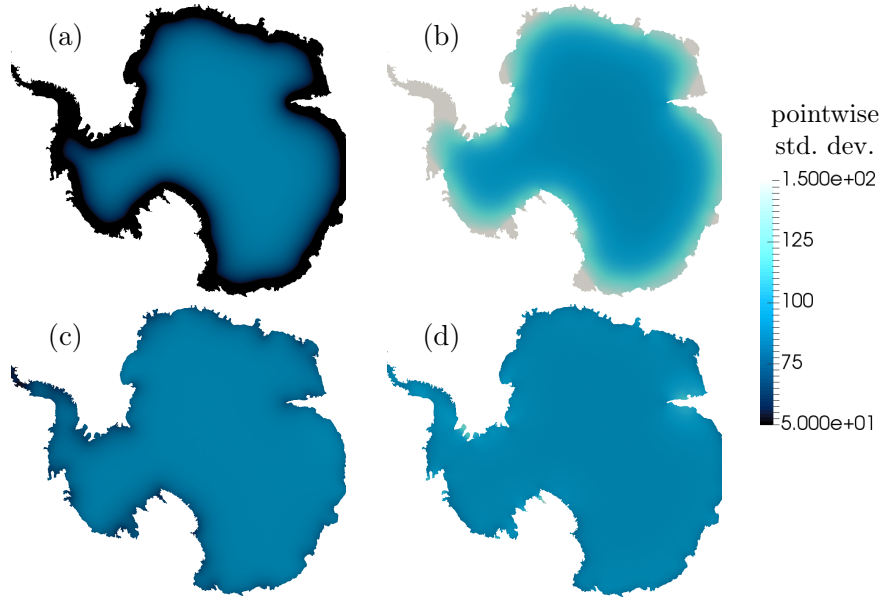
\begin{figure}%[htb]
  \centering
  \begin{tikzpicture}
    \node (11) at (-4,4)
          {\includegraphics[width=0.35\textwidth]{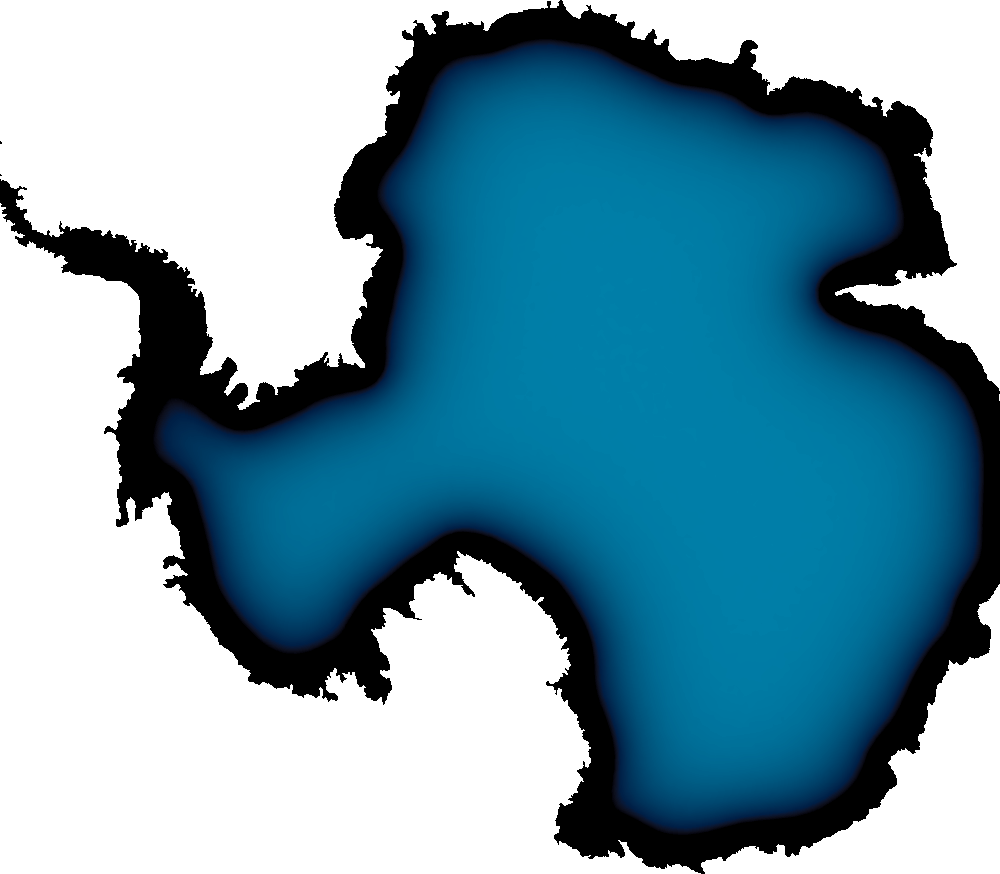}};
    \node (12) at (1,4)
          {\includegraphics[width=0.35\textwidth]{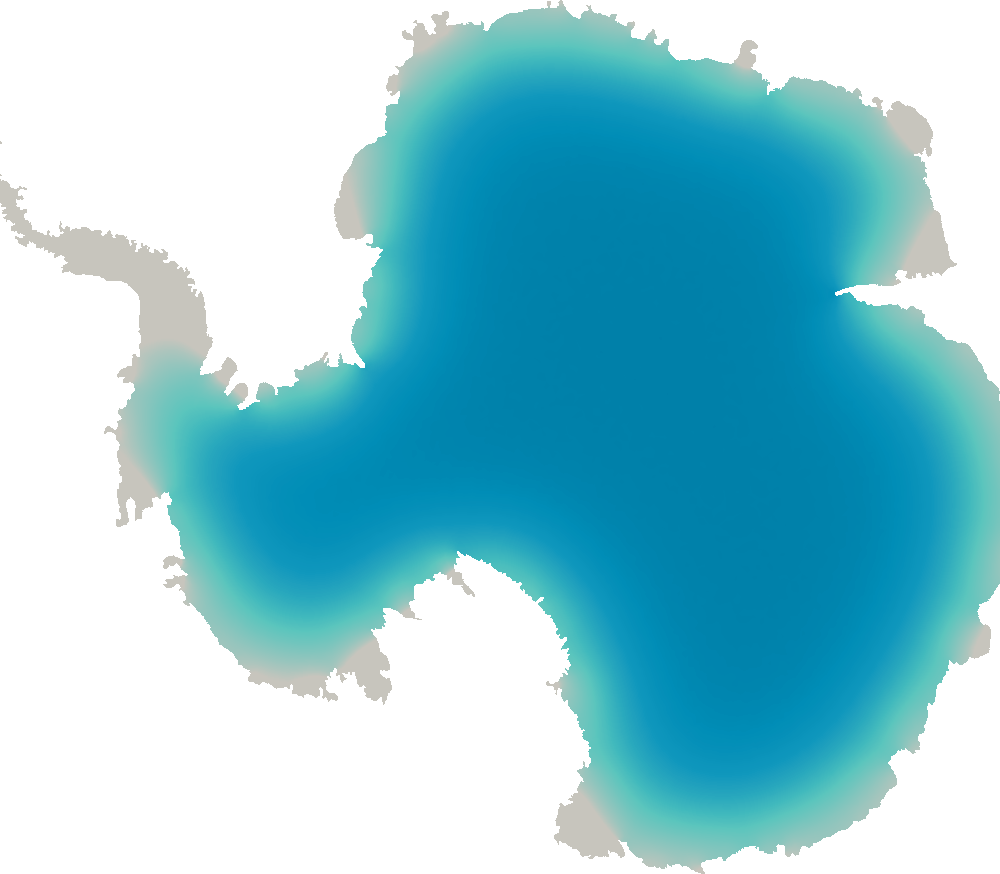}};
    \node (21) at (-4,0)
          {\includegraphics[width=0.35\textwidth]{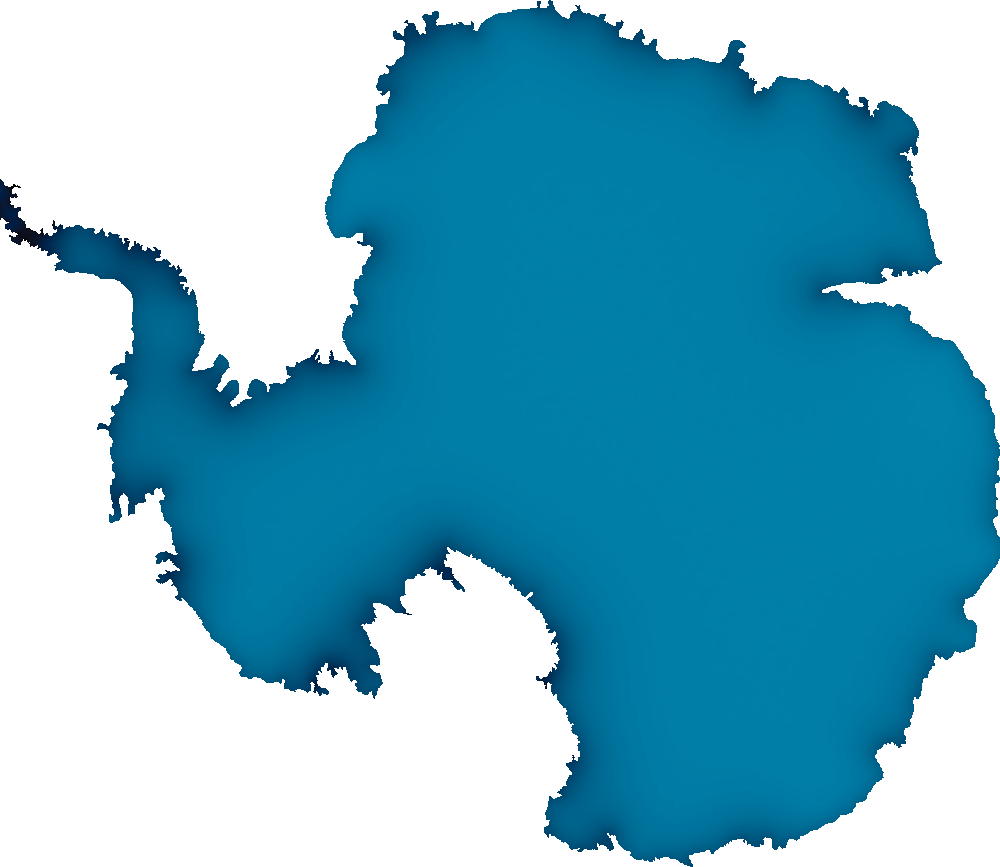}};
    \node (22) at (1,0)
          {\includegraphics[width=0.35\textwidth]{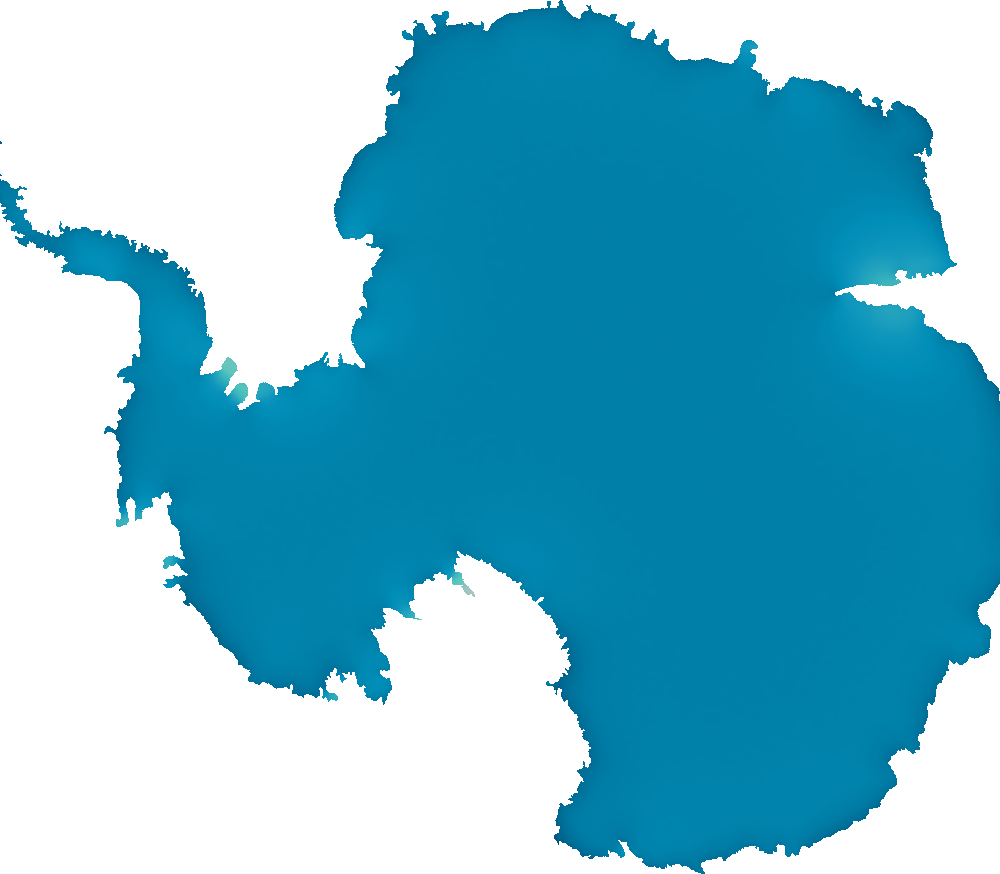}};
    \node (colorbar) at (4.5,2)
          {\includegraphics[width=0.12\textwidth]{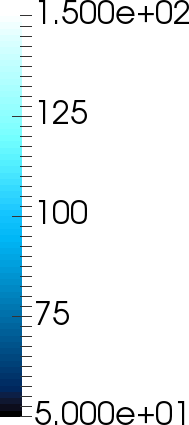}};
          \node at (-5.5,5.5) {\large (a)};
          \node at (-0.5,5.5) {\large (b)};
          \node at (-5.5,1.5) {\large (c)};
          \node at (-0.5,1.5) {\large (d)};
          \node at (4.6,4.25) {%\large
            \begin{tabular}{c}
              pointwise\\
              std.\ dev. 
            \end{tabular}
          };
  \end{tikzpicture}
  \caption{Pointwise standard deviation fields for Antarctica with
    different boundary conditions for the underlying PDE operator:
    Dirichlet conditions (a), Neumann conditions (b), Robin conditions
    with constant coefficient following \mycite{RoininenHuttunenJanneEtAl14} (c),
    and Robin conditions with varying
    coefficient computed as in section \ref{subsec:beta} (d).}
  \label{fig:antarctica variance}
\end{figure}

We also show pointwise standard deviation (i.e., the square root of
the pointwise variance) fields in figure~\ref{fig:antarctica
  variance}.
Since the free-space Green's functions are independent of the boundary,
deviation from a constant pointwise standard deviation is an indicator
for the strength of undesired boundary effects.
We only show standard deviation fields for variants of Robin boundary
conditions as the variance normalization methods discussed in section
\ref{section:variance} ensures constant standard deviation for the
resulting operator. We find that using Dirichlet or Neumann
boundary conditions can have a significant effect also on the
pointwise standard deviations fields. These boundary effects are
significantly diminished for the cases with Robin boundary conditions.

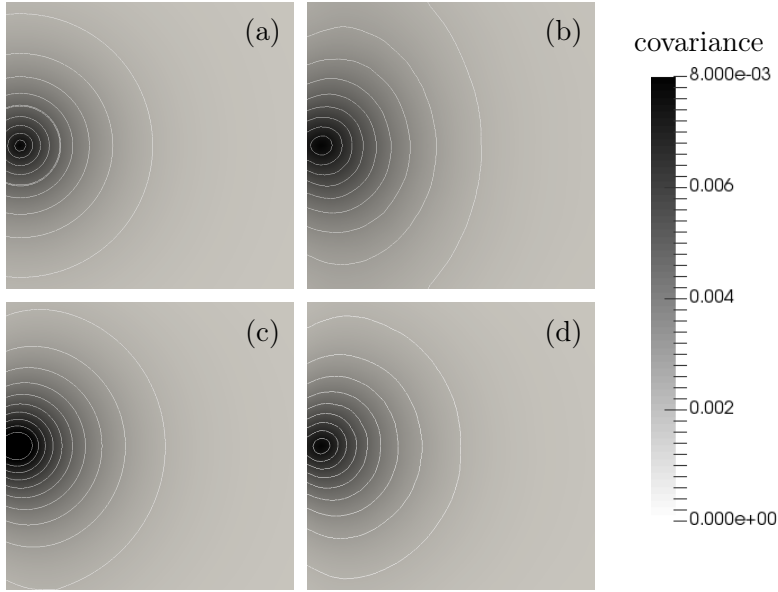
\begin{figure}%[htb]
  \centering
  \begin{tikzpicture}
    
    \node (11) at (-4,4)
    {\includegraphics[width=0.3\textwidth]
      {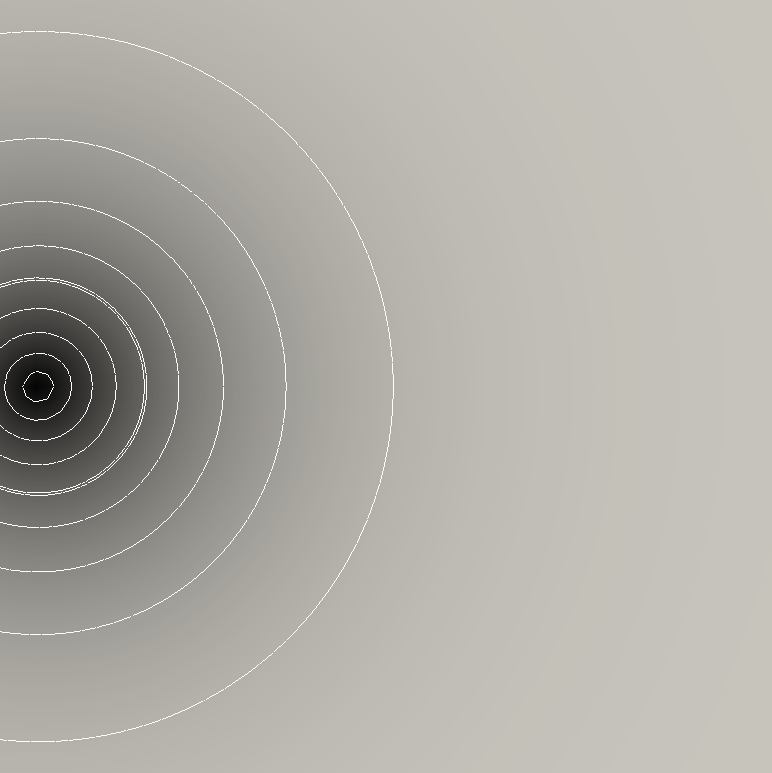}};
    \node at (-2.5,5.5) {\large (a)};
    \node (12) at (0,4)
    {\includegraphics[width=0.3\textwidth]
      {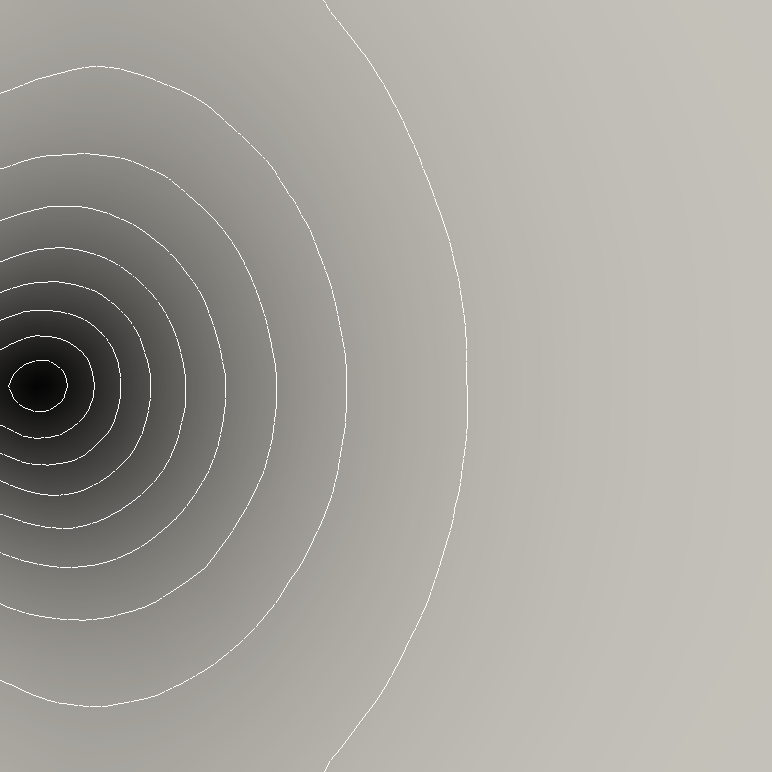}};
    \node at (1.5,5.5) {\large (b)};
    \node (21) at (-4,0)
    {\includegraphics[width=0.3\textwidth]
      {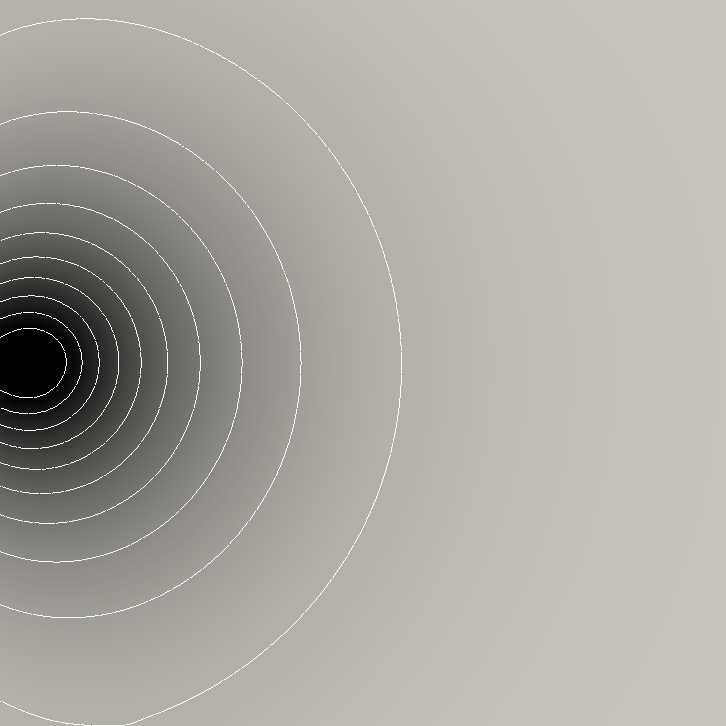}};
    \node at (-2.5,1.5) {\large (c)};
    \node (22) at (0,0)
    {\includegraphics[width=0.3\textwidth]
      {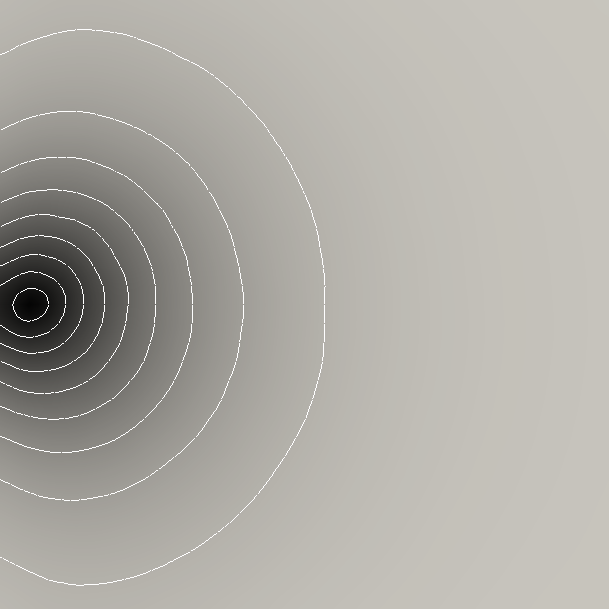}};
    \node at (1.5,1.5) {\large (d)};
    \node (colorbar) at (3.5,2)
    {\includegraphics[width=0.13\textwidth]{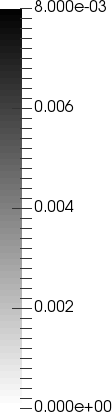}};
    \node at (3.3,5.35) {\large
      \begin{tabular}{c}
        covariance
      \end{tabular}
    };
  \end{tikzpicture}
  \caption{Two-dimensional slices through Green's functions for the
    unit cube example from section \ref{subsec:unit_cube}.
    The center of the green's function is located at $\x^{\star} =
    (0.05,0.5,0.5)^{T}$, and the slice shown is
    $\{ \x^{\star} + (s,0,0)^{T} + (0,t,0)^{T}, s,t \in \mathbb{R} \} \cap [0,1]^3$.
    Shown are the free-space Green's function (a), the Green's
    function computed with Neumann boundary with normalized variance (b),
    with Robin boundary conditions with variable coefficient $\beta$
    (c), and with Robin boundary conditions with
    variable coefficient and normalized variance (d).}
  \label{fig:cube greens}
\end{figure}

\subsection{Unit cube example}\label{subsec:unit_cube}
As three-dimensional test problem, we use the unit cube $[0,1]^3$, and
$\Op = -\Delta + 25$ as the square root of the precision.  We use a
mesh with $64^3$ discretization points. 
In figure~\ref{fig:cube greens}, we show Green's functions for a slice
through the cube.  The values of $\beta$ on a part of
that cross section are shown in figure \ref{fig:beta}. The boundary
conditions shown in figure \ref{fig:cube greens} are all significant
improvements from the results found for either Dirichlet or Neumann
boundary conditions, which we do not show.

\appendix
\section{Statements and proofs for the normalized variance operator}
This section contains precise statements and proofs regarding the
normalizing covariance method presented in section
\ref{section:variance}. We use the notation from section
\ref{section:preliminaries}, and assume a Robin boundary
condition for $\Op$ with $0\le \beta(\x)\le L$, $L>0$, 
for all $\x\in\partial \Omega$. Note that this includes homogeneous
Neumann boundary conditions.
We rely on the assumption that $\Op$ is a
Laplacian-like operator in the sense of \mycite[Assumption
2.9]{Stuart10}. In particular, 
$\Op$ is positive definite, self-adjoint and invertible.

\begin{proposition}[Properties of pointwise variance]\label{prop:g}
  Let $G_2$ be the Green's function of $\Op^{2}$ on a precompact
  domain $\Omega$. Define $g(\x) := {\sigma}/{\sqrt{G_2( \x,\x )}}$,
  with $\sigma > 0$ a constant. Then there exist positive constants $c,C$
  such that $c < g(\x) < C$.
\end{proposition} 
\begin{proof}
  It suffices to show these properties for
  $G_2(\x,\x)$. First, note that
  \begin{equation}\label{eq:convolution}
    G_2(\x,\x) = \int_{\Omega} G_1(\x, \z) G_1( \z, \x ) d\z =
     \int_{\Omega} G_1^2(\x, \z ) d\z.
  \end{equation}

  If there was no positive lower bound for $G_2(\x,\x)$, then due to the
  compactness of $\bar \Omega$, there was $\x \in \bar \Omega$ such
  that $G_2(\x, \x ) = 0$ (we extend Green's functions to $\partial
  \Omega$). However, \eqref{eq:convolution} implies that
  $G_1(\x,\cdot) = 0$ almost everywhere. From the probabilistic
  interpretation of these Green's functions as (density of) time spent
  at a point this can only happen if a particle is immediately killed
  at $\x$. This is not possible for an interior point and can only
  be possible for a boundary point with a homogeneous Dirichlet
  boundary condition, which we exclude. We may conclude that no such
  sequence exists and that there is some lower bound $c > 0$ for which
  $G_2(\x, \x) > c > 0$.

  We know $G_2(\x,\x)$ is the pointwise variance of a $u \sim
  \mathcal{N} (0, \Op^{-2})$ and by the Karhunen-Lo\`eve expansion
  $u(\x) = \sum_{k \in \mathbb{K}} \lambda^{1/2}_k \phi_k(\x) \xi_k$
  with $\xi_k \sim \mathcal{N}(0,1)$ iid and $\{\lambda_k,\phi_k\}_{k \in \mathbb{K}}$
  eigenpairs of the (trace-class) operator $\Op^{-2}$. Then
  \begin{equation}\label{eq:diag}
    G_2(\x,\x) = \var \left (u(\x) \right ) =\mathbb{E}[u^2(\x)] = \sum_{k \in \mathbb{K}} \lambda_k \phi^2_k(\x).  
  \end{equation}
  Since $\Op$ is a Laplacian-like operator according to
  \mycite[Assumption 2.9]{Stuart10} we have a uniform bound on $\|
  \phi_k \|_{\infty}$ and since $\Op^{-2}$ is trace-class $\sum_{k \in
    \mathbb{K}} \lambda_k <\infty$. Using these facts in
  \eqref{eq:diag} gives a uniform upper bound $G_2(\x,\x) < C <
  \infty$, as desired.
\end{proof}

\begin{definition}\label{def:C}
  Let $\covop$ be defined via 
  $[\covop u](\x) := g(\x) [\Op^{-2}(gu)](\x)$
  with $g(\x)$ as in proposition \ref{prop:g}.  
  In the following, we write $\covop = g\Op^{-2} g$, with the
  understanding that $\Op^{-2}$ operates on the product of all 
  functions to its right.
\end{definition}
We now show that $\covop$ is a valid covariance operator with constant pointwise variance.

\begin{proposition}[Properties of the covariance operator]\label{prop:covar}
 $\covop$ is positive definite, self-adjoint, invertible, trace-class
  and has constant pointwise variance. Moreover $u \sim \mathcal{N}(0,
  \covop)$ satisfies $u \sim gv$, where $v \sim
  \mathcal{N}(0,\Op^{-2})$.
\end{proposition}
\begin{proof}
  First observe that:
  \begin{align*}
    (\covop \delta_{\x}) (\x) &= \left ( g \Op^{-2}(g \delta_{\x}) \right )(\x) \\ 
    &= \frac{\sigma}{\sqrt{G_2(\x,\x)}}  (\Op^{-2}( \frac{\sigma}{\sqrt{G_2(\x,\x)}} \delta_{\x}))(\x) \\ 
    &= \frac{\sigma^2}{G_2(\x,\x)}  (\Op^{-2} \delta_{\x})(\x) \\ 
    &= \sigma^2.
  \end{align*}
  Since its diagonal is constant, $\covop$ is trace class with trace
  \begin{equation*}
    Tr( \covop ) = \mathbb{E}_{u\sim \mathcal{N} (0, \covop)} \|u
    \|_2^2 =  \sigma^2 |\Omega|
  \end{equation*}
  (the first equality follows from the Karhunen-Lo\`eve expansion).
  Let $u \in L^2(\Omega)$. By proposition \ref{prop:g}, $ug \in
  L^2(\Omega)$.  Then positive definiteness follows from the fact that
  $\Op^{-2}$ is positive definite. A straightforward calculation shows
  that $\covop$ is self-adjoint in the $L^2(\Omega)$ inner product.
  Using $g > 0$ from proposition \ref{prop:g} and the fact that $\Op$
  is invertible on $\dom(\Op)$, it is easy to verify that
  \begin{align*}%\label{inverse}
    \covop^{-1} &= g^{-1} \Op^{2}g^{-1}. 
  \end{align*}
  The last statement follows from \mycite[Proposition 1.18]{Prato06}.
\end{proof}

\section{Explicit expressions for varying Robin coefficients}\label{subsec:explicit}
Here, we give explicit expressions for the Robin coefficient
function $\beta$ from section
\ref{section:robin}. Let us first recall the  expressions
for the free-space Green's functions for $\Op$ and $\Op^2$.
For an elliptic differential operator $\precop$,
the free-space Green's function is defined (informally) as the
solution to the equation
\begin{align*}
  \precop \Phi(\x,\y) &= \delta_{\x}(\y), \forall \x,\y\in \mathbb{R}^d,
\end{align*}
where $\precop$ operates in $\y$. 

Recall that for a fixed $\x \in \mathbb{R}^d$, $\Phi_p$ satisfies
\begin{equation*}
\Op^p \Phi_p(\x, \cdot) =
(-\gamma \Delta + \alpha)^p \Phi_p(\x, \cdot) =
\gamma^{p} (-\Delta + \alpha / \gamma )^{p} \Phi_p(\x, \cdot) =
\delta_{\x}.
\end{equation*}
Denote $\kappa = \sqrt{ \alpha / \gamma}$. We see that $\Phi_p( \x,
\cdot ) = \gamma^{-p} (-\Delta + \kappa^2)^{-p}\delta_{\x}$. Now we
can recover $\Phi_p$ from known formulas. The following equalities for
the fundamental solutions to Helmholtz (sometimes called screened
Poisson) equations can be verified by differentiation, equation
\eqref{eq:matern expression} with $\nu = 0$ and taking the Laplacian
in polar coordinates, respectively:
\begin{align*}
  d = 1 &\Rightarrow \Phi_1( x, y) =  \frac{\exp( -\kappa |x - y| )}{2\kappa \gamma}, \\
  d = 2 &\Rightarrow \Phi_1(\x,\y) = \frac{1}{2\pi\gamma} K_{0}( \kappa \|\x-\y\|),\\
  d = 3 &\Rightarrow \Phi_1(\x,\y) = \frac{\kappa}{4\pi\gamma} 
  \frac{ \exp( -\kappa\|\x-\y\| )}{ \kappa \|\x-\y\|},
\end{align*}
where $K_\nu$ is the modified Bessel function of the second kind of
order $\nu \in \mathbb{R}$.  For higher powers of $\Op$, let
$\nu := p - d/2$
the free-space Green's function of $\Op^p = \Op^{ \nu + d/2 }$
is the Mat\'ern covariance function \mycite{LindgrenRueLindstroem11,
  Whittle63}:
\begin{align}\label{eq:matern expression}
\Phi_p( \x, \y )  
&= \frac{\sigma^2}{2^{\nu-1}\Gamma(\nu)} (\kappa\|\x-\y\|)^{\nu} K_{\nu}( \kappa\|\x-\y\|),\ p > 1, 
\end{align}
with
\begin{align}\label{eq:sig2}
  \sigma^2 
  &= \frac{\Gamma(\nu)}{\Gamma(\nu + d/2) (4\pi)^{d/2} \kappa^{2\nu} \gamma^{\nu + d/2}} \\
  &= \frac{\Gamma(\nu)}{\Gamma(\nu + d/2) (4\pi)^{d/2} \alpha^{\nu} \gamma^{d/2}}.
\end{align}
Using \mycite[eq.\ 10.30.2]{NIST:DLMF}, it can be verified that
$\Phi_p( \bs 0, \bs 0 ) = \sigma^2$, i.e.,
$\sigma^2$ is the pointwise variance of a Mat\'ern field.

Below, we present $\tilde\beta$ 
from equation \eqref{eq:beta}.
We use the facts that $K_{-\nu} = K_{\nu}$ and 
$(z^{\nu}K_{\nu}(z))' = -z^{\nu}K_{\nu-1}(z)$ \mycite[10.27.4, 10.29.4]{NIST:DLMF}.
We denote $r := \xmy$ and note
that all occurring Green's functions $\Phi_p, p=1,2$ only depend on $\kr$.
Thus, $\frac{\partial \Phi_p}{\dn} = \Phi_p'(\kr) \frac{\partial \kr}{\dn}$ with 
$\frac{\partial r}{\dn } =  \frac{ (\y-\x) \cdot \n }{ r }$,
where $\n$ is the outward pointing unit vector normal at $\y\in \partial \Omega$.
Note that all prefactors in $\Phi_p,p=1,2$ (i.e.,\ $\frac{ \sigma^2}{\Gamma (\nu ) 2^{\nu-1}}$)
cancel out so we ignore them from the outset.

\begin{align*}
  \tilde{\beta}_{2D}(\y) 
  &=  - \frac{\int_{\Omega}  \Phi_1(\kr)  \Phi_2'(\kr)\frac{\partial \kr }{\dn} +
    \Phi_2(\kr)  \Phi_1'(\kr) \frac{\partial \kr}{\dn} d\x}{2\int_{\Omega}
    \Phi_1(\kr) \Phi_2(\kr) d\x} \\
  &= \frac
  {\kappa \int_{\Omega} \kr[K_{0}^2(\kr) + K_{1}^2(\kr)] 
    \frac{ (\y-\x) \cdot \n }{r} d\x}
  {2\int_{\Omega}  \kr K_{1}(\kr) K_{0}(\kr)  d\x}\\[.5ex]
  &=\frac
  { \kappa\int_{\Omega} [ K_{1}^2(\kr) +  K_{0}^2(\kr) ] (\y-\x) \cdot \n d\x}
  {2\int_{\Omega} r K_{1}(\kr) K_{0}(\kr)  d\x}.
  \numberthis \label{eq:explicit beta 2D} 
\end{align*}
In three dimensions we obtain: 
\begin{align*}
  \tilde{\beta}_{3D}(\y) 
  &=  - \frac{\int_{\Omega}  \Phi_1(\kr) \Phi_2'(\kr) \frac{\partial \kr }{\dn} +
    \Phi_2(\kr) \Phi_1'(\kr) \frac{\partial \kr}{\dn} d\x}
  {2\int_{\Omega} \Phi_1(\kr) \Phi_2(\kr) d\x} \\
  &= \frac{
    \kappa \int_{\Omega} 
    [(\kr)^{-1}e^{-\kr}  \sqrt{\kr} K_{-\frac12}(\kr) + 
    \sqrt{\kr} K_{\frac12}(\kr) (e^{-\kr}\kr +  e^{-\kr}) (\kr)^{-2}]
    \frac{\partial r}{\dn} d\x 
  }{
    2\int_{\Omega}  \sqrt{\kr} K_{\frac12}(\kr) e^{-\kr} (\kr)^{-1}  d\x
  } \\
  &=  \frac{ 
    \kappa \int_{\Omega}  
    [ e^{-\kr} /\sqrt{\kr} K_{\frac12}(\kr) +  K_{\frac12}(\kr) (e^{-\kr}\kr +  e^{-\kr})(\kr)^{-\frac32}]
    \frac{\partial r}{\dn} d\x 
  }{
    2\int_{\Omega}  K_{\frac12}(\kr) e^{-\kr} / \sqrt{\kr}   d\x
  } \\
  &=  \frac{ \kappa
    \int_{\Omega} e^{-\kr} / \sqrt{\kr} K_{\frac12}(\kr) 
    [ 1 +  (\kr+1) / \kr ] \frac{\partial r}{\dn} d\x
  }{
    2\int_{\Omega} K_{\frac12}(\kr) e^{-\kr} / \sqrt{\kr} d\x
  } \\
  &= \frac{
    \kappa \int_{\Omega} r^{-\frac32} (2 + \frac{1}{\kr}) e^{-\kr} K_{\frac12}(\kr)  
    (\y-\x) \cdot \n  d\x
  }{
    2\int_{\Omega}  K_{\frac12}(\kr) e^{-\kr} r^{-\frac12} d\x
  }. 
  \numberthis \label{eq:explicit beta 3D}
\end{align*}
\bibliographystyle{unsrt}
\bibliography{georg_refs,refs}

\begin{thebibliography}{10}

\bibitem{Stuart10}
Andrew~M. Stuart.
\newblock Inverse problems: {A B}ayesian perspective.
\newblock {\em Acta Numerica}, 19:451--559, 2010.

\bibitem{Bui-ThanhGhattasMartinEtAl13}
Tan Bui-Thanh, Omar Ghattas, James Martin, and Georg Stadler.
\newblock A computational framework for infinite-dimensional {B}ayesian inverse
  problems {P}art {I}: {T}he linearized case, with application to global
  seismic inversion.
\newblock {\em SIAM Journal on Scientific Computing}, 35(6):A2494--A2523, 2013.

\bibitem{LindgrenRueLindstroem11}
Finn Lindgren, H{\aa}vard Rue, and Johan Lindstr{\"o}m.
\newblock An explicit link between {G}aussian fields and {G}aussian {M}arkov
  random fields: the stochastic partial differential equation approach.
\newblock {\em Journal of the Royal Statistical Society: Series B (Statistical
  Methodology)}, 73(4):423--498, 2011.

\bibitem{RoininenHuttunenJanneEtAl14}
Lassi Roininen, Janne M.~J. Huttunen, and Sari Lasanen.
\newblock Whittle-{M}at{\'e}rn priors for {B}ayesian statistical inversion with
  applications in electrical impedance tomography.
\newblock {\em Inverse Problems Imaging}, 8(2):561--586, 2014.

\bibitem{IsaacPetraStadlerEtAl15}
Tobin Isaac, Noemi Petra, Georg Stadler, and Omar Ghattas.
\newblock Scalable and efficient algorithms for the propagation of uncertainty
  from data through inference to prediction for large-scale problems, with
  application to flow of the {A}ntarctic ice sheet.
\newblock {\em Journal of Computational Physics}, 296:348--368, September 2015.

\bibitem{SimpsonLindgrenFinnEtAl12s}
Daniel Simpson, Finn Lindgren, and H{\aa}vard Rue.
\newblock In order to make spatial statistics computationally feasible, we need
  to forget about the covariance function.
\newblock {\em Environmetrics}, 23(1):65--74, 2012.

\bibitem{SimpsonLindgrenFinnEtAl12}
Daniel Simpson, Finn Lindgren, and H{\aa}vard Rue.
\newblock Think continuous: {M}arkovian {G}aussian models in spatial
  statistics.
\newblock {\em Spatial Statistics}, 1:16--29, 2012.

\bibitem{Whittle63}
Peter Whittle.
\newblock Stochastic-processes in several dimensions.
\newblock {\em Bulletin of the International Statistical Institute},
  40(2):974--994, 1963.

\bibitem{Besag81}
Julian Besag.
\newblock On a system of two-dimensional recurrence equations.
\newblock {\em Journal of the Royal Statistical Society. Series B
  (Methodological)}, pages 302--309, 1981.

\bibitem{CalvettiKaipioSomersalo06}
Daniela Calvetti, Jari~P Kaipio, and Erkki Somersalo.
\newblock Aristotelian prior boundary conditions.
\newblock {\em International Journal of Mathematics and Computer Science},
  1:63--81, 2006.

\bibitem{Hairer09}
Martin Hairer.
\newblock Introduction to {S}tochastic {PDE}s.
\newblock Lecture Notes, 2009.

\bibitem{Prato06}
Giuseppe Da~Prato.
\newblock {\em An Introduction to Infinite-dimensional Analysis}.
\newblock Universitext. Springer, 2006.

\bibitem{Evans10}
Lawrence~C. Evans.
\newblock {\em Partial Differential Equations}.
\newblock Graduate studies in mathematics. American Mathematical Society,
  second edition, 2010.

\bibitem{Cubature}
Steven~G. Johnson.
\newblock {Cubature---Adaptive Multi-dimension Integration}.
\newblock \url{http://ab-initio.mit.edu/wiki/index.php/Cubature}.

\bibitem{RueMartino07}
H{\aa}vard Rue and Sara Martino.
\newblock Approximate {B}ayesian inference for hierarchical {G}aussian {M}arkov
  random field models.
\newblock {\em Journal of statistical planning and inference},
  137(10):3177--3192, 2007.

\bibitem{LinLuYingEtAl09}
Lin Lin, Jianfeng Lu, Lexing Ying, Roberto Car, and Weinan E.
\newblock Fast algorithm for extracting the diagonal of the inverse matrix with
  application to the electronic structure analysis of metallic systems.
\newblock {\em Communications in Mathematical Sciences}, 7(3):755--777, 2009.

\bibitem{BekasCurioniFedulova09}
Costas Bekas, Alessandro Curioni, and Irina Fedulova.
\newblock Low cost high performance uncertainty quantification.
\newblock In {\em Proceedings of the 2nd Workshop on High Performance
  Computational Finance}, WHPCF '09, pages 8:1--8:8, New York, NY, USA, 2009.
  ACM.

\bibitem{BekasKokiopoulouSaad07}
Costas Bekas, Effrosyni Kokiopoulou, and Yousef Saad.
\newblock An estimator for the diagonal of a matrix.
\newblock {\em Applied Numerical Mathematics}, 57(11):1214--1229, 2007.

\bibitem{TangSaad12}
Jok~M. Tang and Yousef Saad.
\newblock A probing method for computing the diagonal of a matrix inverse.
\newblock {\em Numerical Linear Algebra with Applications}, 19(3):485--501,
  2012.

\bibitem{LoggMardalWells12}
Anders Logg, Kent-Andre Mardal, and Garth~N. Wells, editors.
\newblock {\em Automated Solution of Differential Equations by the Finite
  Element Method}, volume~84 of {\em Lecture Notes in Computational Science and
  Engineering}.
\newblock Springer, 2012.

\bibitem{NIST:DLMF}
{NIST Digital Library of Mathematical Functions}.
\newblock \url{http://dlmf.nist.gov/}, Release 1.0.10 of 2015-08-07.
\newblock Online companion to \cite{Olver:2010:NHMF}.

\bibitem{Olver:2010:NHMF}
F.~W.~J. Olver, D.~W. Lozier, R.~F. Boisvert, and C.~W. Clark, editors.
\newblock {\em {NIST Handbook of Mathematical Functions}}.
\newblock Cambridge University Press, New York, NY, 2010.
\newblock Print companion to \cite{NIST:DLMF}.

\end{thebibliography}

\end{document}